\documentclass{article}
\usepackage[dvips]{graphicx}
\usepackage{float}
\usepackage{latexsym}
\usepackage{color,amsmath,mathrsfs,bm,amssymb,color}
\newtheorem{theorem}{Theorem}[section]
\newtheorem{definition}[theorem]{Definition}
\newtheorem{lemma}[theorem]{Lemma}
\newtheorem{proposition}[theorem]{Proposition}

\newtheorem{remark}[theorem]{Remark}
\newtheorem{example}[theorem]{Example}
\newcommand{\proof}{Proof:\ }
\newcommand{\qed}{\hfill $\Box$ \par\medskip}

\catcode`\@=11
\def\newsymbol#1#2#3#4#5{\let\next@\relax%
 \ifnum#2=\@ne\else%
 \ifnum#2=\tw@\let\next@\msyfam@\fi\fi%
 \mathchardef#1="#3\next@#4#5}
\def\mathhexbox@#1#2#3{\relax%
 \ifmmode\mathpalette{} {\m@th\mnnathchar"#1#2#3}
 \else\leavevmode\hbox{$\m@th\mathchar"#1#2#3$}\fi}
\font\tenmsy=msbm10
\font\sevenmsy=msbm7
\font\fivemsy=msbm5
\newfam\msyfam
\textfont\msyfam=\tenmsy
\scriptfont\msyfam=\sevenmsy
\scriptscriptfont\msyfam=\fivemsy
\edef\msyfam@{\hexnumber@\msyfam}

\catcode`\@=\active

\newcommand{\B}{{\mathcal{B}}}
\newcommand{\C}{{\mathcal{C}}}
\newcommand{\W}{{\mathcal {W}}}
\newcommand{\X}{{\ms X}}
\newcommand{\D}{\tilde{\X}}

\newcommand{\XX}{{{\mathfrak X}}}
\newcommand{\DD}{{\tilde{\mathfrak X}}}

\newcommand{\CC}{{{\mathbb C}}}
\newcommand{\RR}{{\mathbb R}}

\newcommand{\XXQ}{\XX_\hir}

\newcommand{\Spec}{{\rm Spec}}

\newcommand{\WIS}{{U_{\rm f}}}

\newcommand{\ppp}{{\tilde{\W}}}
\newcommand{\Y}{\XX_Q}
\newcommand{\F}{\mathcal {F}}
\newcommand{\N}{{\mathcal {N}}}
\renewcommand{\P}{{\mathcal {P}}}
\newcommand{\Q}{{\mathcal {Q}}}
\newcommand{\M}{{\mathcal {M}}}
\newcommand{\G}{{\mathcal {G}}}
\newcommand{\LN}{L_{\rm N}}

\newcommand{\xx}{b}

\newcommand{\eq}[1]{\begin{equation}\label{#1}}
\newcommand{\en}{\end{equation}}

\newcommand{\BR}{{\RR^d}}
\newcommand{\Ebb}{{\mathbb E}}
\newcommand{\TT}{J_\eps}
\newtheorem{assumption}{Assumption}

\newcommand{\TTT}[1]{{\bf (#1)}}

\newcommand{\la}{\lambda}
\newcommand{\eps}{\varepsilon}
\newcommand{\bi}{\begin{description}}
\newcommand{\ei}{\end{description} }
\newcommand{\bd}[1]{\begin{definition}\label{#1}}
\newcommand{\ed}{\end{definition}}

\newcommand{\bl}[1]{\begin{lemma}\label{#1}}
\newcommand{\el}{\end{lemma}}
\newcommand{\bc}[1]{\begin{corollary}\label{#1}}
\newcommand{\ec}{\end{corollary}}
\newcommand{\bt}[1]{\begin{theorem}\label{#1}}
\newcommand{\et}{\end{theorem}}
\newcommand{\bp}[1]{\begin{proposition}\label{#1}}
\newcommand{\ep}{\end{proposition}}
\newcommand{\br}[1]{\begin{remark}\label{#1}}
\newcommand{\er}{\end{remark}}

\newcommand{\kak}[1]{(\ref{#1})}
\newcommand{\LM}{{L^2({\rm M})}}
\newcommand{\LP}{{L^2({\rm P})}}
\newcommand{\LR}{{L^2(\BR)}}
\newcommand{\K}{{\mathcal  K}}
\newcommand{\KKK}{\varphi_{\rm pg}}
\newcommand{\fff}{{\ffff_{\rm b}}}
\newcommand{\ffff}{\mathscr{F}}

\newcommand{\FF}[1]{\tilde F_{#1}(f)}
\newcommand{\FFF}[1]{F_{#1}(f)}

\newcommand{\ms}[1]{\mathscr{#1}}
\newcommand{\is}{\inf{\Spec}}
\newcommand{\f}{^{-1}}
\newcommand{\lk}{\left(}
\newcommand{\rk}{\right)}
\newcommand{\lkk}{\left\{}
\newcommand{\rkk}{\right\}}

\newcommand{\pro}[1]{(#1_t)_{t\in\RR}}

\newcommand{\Ree}{\mbox{Re} }

\newcommand{\add}{a^{\ast}}

\newcommand{\hp}{H_{\rm p}}
\newcommand{\hn}{H_{\rm N}}

\newcommand{\lp}{L_{\rm p}}

\newcommand{\ov}[1]{\overline{#1}}
\newcommand{\hf}{H_{\rm f}}

\newcommand{\gr}{\varphi_{\rm g}}
\newcommand{\ggr}{\tilde \varphi_{\rm g}}

\newcommand{\grp}{\varphi_{\rm p}}

\newcommand{\half}{\frac{1}{2}}
\newcommand{\han}{{1/2}}

\newcommand{\hir}{{\mathbb M}}
\newcommand{\hirr}{{\hir^{[0,\infty)}}}

\newcommand{\hi}{H_{\rm i}}

\def\bbbone{{\mathchoice {\rm 1\mskip-4mu l} {\rm 1\mskip-4mu l}
{\rm 1\mskip-4.5mu l} {\rm 1\mskip-5mu l}}}
\def\one{\bbbone}

\renewcommand{\SS}{{\mathbb S}}
\newcommand{\s}{\sigma}
\newcommand{\hhh}{{\mathscr{H}}}

\renewcommand{\d}{\displaystyle}
\newcommand{\vp}{{\hat  \varphi}}
\newcommand{\vpp}{{\tilde \varphi}}

\newcommand{\non}{\nonumber}

\newcommand{\proo}[1]{(#1_t)_{t\geq 0}}

\providecommand{\seq}[1]{(#1_n)_{n\in \mathbb{N}}}

\begin{document}

\makeatletter
\renewcommand\@dotsep{10000}
\makeatother


\title{Functional central limit theorems and $P(\phi)_{1}$-processes for the relativistic  and non-relativistic Nelson models }
\author{
Soumaya Gheryani, \thanks{College of Science, Jazan University Saudi Arabia, sgheryani@jazanu.edu.sa}
Fumio Hiroshima, \thanks{Faculty of Mathematics, Kyushu University, Fukuoka, Japan, hiroshima@math.kyushu-u.ac.jp}
J\'ozsef L\H{o}rinczi, \thanks{Department of Mathematical Sciences, Loughborough University, UK, J.Lorinczi@lboro.ac.uk}\\ 
Achref Majid,\thanks{Ecole Sup\'erieure Priv\'ee d'Ing\'enieurie et des Technologies (ESPRIT), Tunisie, achref.lemjid@esprit.tn}
Habib Ouerdiane \thanks{Faculty of Science, University of Tunis El Manar, Tunisia, habib.ouerdiane@fst.rnu.tn}
}
\maketitle
\begin{abstract}
We construct $P(\phi)_1$-processes indexed by the full time-line, separately derived from the functional integral
representations of the relativistic and non-relativistic Nelson models in quantum field theory. These two cases
differ essentially by sample path regularity. Associated with these processes we define a martingale which, under
an appropriate scaling, allows to obtain a central limit theorem for additive functionals of these processes. We
discuss a number of examples by choosing specific functionals related to particle-field operators.
\bigskip
\noindent


\medskip
\noindent
2010 MS Classification:  47D07, 47D08, 60J75, 81Q10

\end{abstract}

\section{Introduction}
\subsection{Definition of the non-relativistic and relativistic Nelson Hamiltonian}
Nelson's model of an electrically charged spinless quantum mechanical particle coupled to a scalar boson field \cite{nel64a,nel64b}
has been much studied recently in rigorous quantum field theory. 
In this paper we are interested in 
constructing Markov  processes associated with 
both the relativistic and non-relativistic Nelson models, and 
studying some properties of both models. 
These properties will be formulated in terms of central limit theorem-type behaviours of functionals of the particle-field operators. 
Our main contribution is a technical development to quantum field theory of L\'evy or
L\'evy-type processes in infinite dimensions which, to the best of our knowledge, has not been attempted before.

The Nelson Hamiltonian is defined by a self-adjoint operator of the form
\begin{equation}
\label{nelsonH}
\hn  = \hp  \otimes \one  + \one  \otimes \hf  + \hi,
\end{equation}
on the Hilbert space 
$$\hhh=\LR\otimes \fff,$$
where $\fff= \oplus_{n=0}^\infty L^2_{\rm sym}(\RR^{dn})$ denotes the boson Fock space over $\LR$. 
Here $\fff^{(n)}=  L^2_{\rm sym}(\RR^{dn})$. 
The Fock space can be identified with the space of $\ell_2$-sequences
$(\psi^{(n)})_{n\in \mathbb{N}}$ such that $\psi^{(n)}\in  \fff^{(n)}$ and
$\|\psi\|_{\fff}^2=\sum_{n=0}^{\infty}  \|(\psi^{(n)})\|^2_{\fff^{(n)}}<\infty$. 
We denote the annihilation and creation operators by $a(f)$ and $\add (f)$, $f\in \LR $,
respectively, satisfying the canonical commutation relations
$$[a(f),\add (g)]=(\bar{f},g)\one,\quad 
[a(f),a(g)]=0=[\add (f),\add (g)]$$
on a dense domain of  $\fff$. 

The components of $\hn$  describe
the Hamilton operators of the quantum mechanical  
particle, free field, and particle-field interaction, respectively.

Let $\varphi:\BR \rightarrow\RR $ be a function describing the charge distribution of the particle, denote by $\vp$
its Fourier transform, and write $\widetilde{\hat  {\varphi}}(k)=\hat  {\varphi}(-k)$. For every $x\in\BR$, define
\eq{hi}
\hi (x)=\frac{1}{\sqrt{2}}\left( \add ({\hat  {\varphi}e^{-ikx}}/{\sqrt{\omega}})+
a({\widetilde{\hat  {\varphi}}e^{ikx}}/{\sqrt{\omega}})\right).
\en
We define the {interaction Hamiltonian} $\hi : \hhh \rightarrow\hhh$ by 
the constant fiber direct integral 
$$(\hi \Psi)(x)=\hi (x)\Psi(x)$$ for $\Psi\in \hhh $ such that $\Psi(x)\in D(\hi(x))$, for almost every
$x\in\BR$. Here we use the identification
$$\hhh\cong \int_\BR^\oplus \fff dx.$$
Formally, this operator is then written as
\eq{e}
\hi (x)=\int_\BR\frac{1}{\sqrt{2\omega(k)}}\left(\hat  {\varphi}(k)e^{-ik\cdot x}\add (k)
+ {\hat  {\varphi}(-k)}e^{ik\cdot x}a(k)\right)dk,
\en
where $a^*(k)$ and $a(k)$ are formal kernels of the creation and annihilation operators, respectively.

Next denote by $$d\Gamma(\omega):\fff\to\fff$$ the second quantization of 
the multiplication operator 
$$\omega(k)=\sqrt{|k|^{2}+\nu^{2}},\quad \nu\geq0$$
acting in $\LR$, defined by
$$d\Gamma(\omega)\Phi^{(n)}(k_1,\ldots,k_n)=(\sum_{j=1}^n \omega(k_j))\Phi^{(n)}(k_1,\ldots,k_n).$$ 
Here $\omega$ denotes the dispersion relation with $\nu$ giving the mass of a single boson. 
We set 
$$\hf=d\Gamma(\omega).$$
Formally, the free field Hamiltonian can be
written as 
\begin{align}\label{hf1}
\hf =\int_\BR \omega(k) \add (k)a(k)dk.
\end{align}

The non-relativistic  and relativistic Nelson Hamiltonians differ by the choice of  particle Hamiltonian $\hp$. Consider the operator
$$
\hp = (-\hbar^2 c^2 \Delta + m^2c^4)^{1/2} - mc^2 + V ,
$$
on $L^2(\RR^d)$, where $V$ is a multiplication operator describing the potential, $m \geq 0$ is the rest mass of the particle,
and $\hbar, c$ are physical constants (Planck's constant and the speed of light, respectively). 
We put the kinetic term by 
$$\hp^0=(-\hbar^2 c^2 \Delta + m^2c^4)^{1/2}.$$
This operator has been much
studied 
 in mathematics as the (semi-) relativistic operator \cite{BE11,cms,H77,W74}. 
We will consider the following variants of this operator in
this paper:
\begin{description}
 \item[\bf (massive relativistic case)]
in the case $m > 0$ and in units in which $c = \hbar = 1$ we have the {massive relativistic Schr\"odinger operator}
\begin{equation}
\label{rela}
\hp = \hp^0 + V \quad \mbox{with} \quad \hp^0 = (-\Delta + m^2)^{1/2} - m
\end{equation}
\item[\bf (massless relativistic case)]
in the case $m=0$ and in units in which $c = \hbar = 1$ we have the {massless relativistic Schr\"odinger operator}
\begin{equation}
\label{masslos}
\hp = \hp^0 + V \quad \mbox{with} \quad \hp^0 = (-\Delta)^{1/2}
\end{equation}
\item[\bf (non-relativistic case)]
in the limit $c \to \infty$ and in units in which $\hbar = m = 1$ we have the {non-relativistic Schr\"odinger
operator}
\begin{equation}
\label{classi}
\hp = \hp^0 + V \quad \mbox{with} \quad \hp^0 = -\frac{1}{2}   \Delta.
\end{equation}
\end{description}

When $\hp$ is given by \eqref{rela} or \eqref{masslos}, we say that the operator $\hn $ is the
{relativistic Nelson Hamiltonian}, and for \eqref{classi} we call it {non-relativistic  Nelson Hamiltonian}.
Let $\Spec (K)$ be the spectrum of self-adjoint operator $K$. 
We will use the following standing assumptions throughout this paper.

\begin{assumption} \label{assumption}
{\rm
The following conditions hold:
\begin{description}
\item[(i)]
\label{r1}$\overline{\hat  {\varphi}(k)}=\hat  {\varphi}(-k)$
and
$\hat  {\varphi}/\sqrt{\omega},
\hat  {\varphi}/\omega\in L^2(\mathbb{R}^d)$.
 \item[(ii)]
The potential $V$ is of Kato-class in the sense of Definition \ref{kato} below.
 \item[(iii)]
For any of the three choices \eqref{rela}-\eqref{classi}, the operator $\hp $ has a unique, strictly positive
ground state $\grp \in D(H_{\rm p})$, with $H_{\rm p} \grp =E_{\rm p}\grp $, $\|\grp\|_{L^2({\BR})} = 1$, 
where 
$$E_{\rm p}=\inf \Spec({\hp}).$$
 \item[(iv)]
For any of the three choices \eqref{rela}-\eqref{classi}, the operator $\hn $ has a unique, strictly positive
ground state $\gr \in D(\hn )$ with $\hn \gr=E\gr$, $\|\gr\|_{\hhh} = 1$, and 
$\||x|^{d+1}\gr\|<\infty$,  
 where 
 $$E=\inf \Spec(\hn ).$$
\end{description}
}
\end{assumption}

\begin{remark}
Using (i) of Assumption \ref{assumption} it is seen that $\hi$ is symmetric, and infinitesimally small with respect to $\hp\otimes\one+\one\otimes\hf$,  
thus $\hn $ is self-adjoint on $D(\hp\otimes\one+\one\otimes\hf)$. 
\end{remark}

\begin{remark}
If the  infrared regular condition $\int_\BR|\vp(k)|^2/\omega(k)^3 dk<\infty$ is satisfied,
then the existence of the ground state of $\hn$ holds. 
This is shown in e.g.,  
\cite{bfs2,ge,ghps2,hhs,spo98}. 
\end{remark}

For the non-relativistic Nelson Hamiltonian 
various aspects concerning  spectral properties are known. 
IR divergences are studied in \cite{A01,ghps3,LMS02a,LMS02b},  
UV divergence in \cite{ghps4,GHL14,HM18, MM17,nel64a,nel64b} and a Gibbs measure associated with the ground state  in \cite{B02}. 
So far, most papers have
focused on the non-relativistic  Nelson model; 
for a monograph-length discussion see also \cite{LHB}.

\begin{remark}\label{X}
We give a remark on the spatial decay of the ground state $\gr$. 
In the case of \eqref{rela} and \eqref{classi} 
it is known that 
\eq{ex}
\|(e^{\beta|x|}\otimes \one) \gr\|<\infty
\en
is satisfied 
when (i) $\underset {|x|\to\infty}  \lim V(x)=\infty$ or  
(ii)~$\underset{|x|\to\infty} \lim V_{-}(x)-E>0$, where $V_{-}(x)=-\min\{0,V(x)\}$. 
We refer to \cite{LHB}. 
In the case of \eqref{masslos},  
when $\underset{|x|\to\infty} \lim V(x)=\infty$, 
\kak{ex} holds true, but 
when $\underset {|x|\to\infty} \lim V(x)-E>0$, 
$$\|(|x|^{d+1}\otimes \one) \gr\|<\infty$$  
is satisfied. See \cite[Theorem 5.15]{hir14} and \cite{cms}. 
\end{remark}

\begin{remark}\label{NN}
Let 
$
\phi(f)=\frac{1}{\sqrt 2}\int_\BR (\add(k) \hat f(k)+a(k)\hat f(-k)) dk
$
 be the  field operator. 
It can be seen in a similar way to \cite{hhl14} that 
$\|e^{\beta \phi(f)^2}\gr\| <\infty$ for $\beta<1/(2\|f\|)$. 
\end{remark}

\subsection{Main results}
Although the models discussed here are defined in terms of self-adjoint operators on a joint particle-field space of functions, for our purposes a Feynman-Kac type approach will be more suitable. Then the related evolution semigroups can
be represented in terms of averages over 
suitable random processes. As it will be seen below in detail, in
the relativistic case we have an infinite-dimensional jump process, and in the non-relativistic case an infinite-dimensional
diffusion.
\begin{definition}[Ground state transform]
Let $K$ be a self-adjoint operator acting in $L^2(M,dm)$ with a measurable space $(M,m)$ and $\psi$ the  strictly positive ground state of $K$ with $K\psi=\kappa \psi$, where $\kappa=\is(K)$. 
Let $M_f$ denote the multiplication by $f$. 
We call the map 
$$K\mapsto L_K=M_{\frac{1}{\psi}}(K-\kappa)M_{\psi}$$ the ground state transform of $K$. 
\end{definition}
For notational convenience 
we write $L_K$ simply as 
$$\frac{1}{\psi}(K-\kappa){\psi}$$ in what follows. 
In  probability theory this transform is usually called Doob's $h$-transform. 
We see that $L_K$ is also self-adjoint on 
$D(L_K)=\{f\in L^2(M,\psi^2dm)\mid f\psi\in D(K)\}$. 
Informally a Markov  process possessing  $K$ as a generator is called 
a $P(\phi)_1$-process. 
The exact definition of  $P(\phi)_1$-processes is given in Definition \ref{matu}. 

Now we explain probabilistic notations used in this paper. 
Let $(X,{\mathcal  B}(X),m)$ be a probability space. The expectation value of measurable function $f:X\to \CC$ with respect to $m$ is denoted by $\Ebb_m[f]$, i.e., 
$$\Ebb_m[f]=\int_X f(x) dm.$$ 
Let $(q_t)_{t\geq0}$ be an $\RR^d$-valued random process on 
$(X,{\mathcal  B}(X),m^x)$, where $m^x$,$x\in \RR^d$,  is a probability measure such that $m^x(q_0=x)=1$.
In this case we often  use  notation $\Ebb_m^x$ instead of 
$\Ebb_{m^x}$.   Furthermore we write $\Ebb^x$ for $\Ebb_m^x$ unless confusion may arise.

We present a simple example of  a $P(\phi)_1$-process. 
The Ornstein-Uhlenbeck process is a $P(\phi)_1$-process. 
Let $h=-\frac{1}{2}\Delta+\half x^2$ be the one-dimensional harmonic oscillator defined on $L^2(\RR)$. 
We  see that $h$ has the strictly positive ground state $\psi(x)=\pi^{-\frac{1}{4}}e^{-\frac{x^2}{2}}$ satisfying 
$h\psi=\half  \psi$.  
The ground state transform of $h$ is given by 
 $$L_h=\frac{1}{\psi} (h-\half)\psi.$$ 
 The operator $L_h$ is self-adjoint  
acting in $L^2(\RR,\psi^2 dx)$
 and notice that $\psi^2 dx $ is 
a probability measure on $\RR$.
Also, $L_h\one=0$. 
The Ornstein-Uhlenbeck process $(y_t)_{t\geq0}$ is a diffusion process on a probability space 
and fullfils  
$$(f, e^{-tL_h} g)_{ L^2(\RR,\psi^2 dx)}=\int_{\RR} \Ebb^x[\bar f(y_0) g(y_t)]\psi^2(x) dx.$$
We can see that 
$$L_hf=
-\half \Delta f+x \nabla f
$$ and 
$(y_t)_{t\geq0}$ is a solution of SDE
$$dy_t=y_t dt +dB_t,$$
where $(B_t)_{t\geq0}$ is  Brownian motion. 

This can be extended to more general $d$-dimensional Schr\"odinger operators $\hp$ in \cite{LHB,LO12}. 
Let $$L_{\rm p}=\frac{1}{\grp}(\hp-E_{\rm p})\grp$$ and 
the non-relativistic case 
formally
we can see that 
$$L_{\rm p}f=-\half \Delta f-\frac{\nabla \grp}{\grp}\cdot \nabla f.$$ 
We denote by $(z_t)_{t\geq0}$ the $P(\phi)_1$-process associated with $\hp$ on a
probability space.  
$(z_t)_{t\geq0}$  is a possible solution of SDE
\begin{align}\label{SDE}
dz_t=\frac{\nabla \grp}{\grp}(z_t)dt +dB_t.
\end{align}
It is however not trivial to see the regularity of $\grp$ and the existence of the solution of 
\eqref{SDE}.
Nevertheless 
in \cite{LHB} 
 a $P(\phi)_1$-process  associated with 
a  non-relativistic Schr\"odinger operator $\hp $  is constructed and 
 in \cite{LO12} 
that for a relativistic case. 
This implies that 
\begin{align}
\label{FKF}
&(f, e^{-tL_{\rm p}} g)_{ L^2(\BR,\varphi_{\rm p}^2dx)}=\int_{\BR} 
\Ebb^x[\bar f(z_0) g(z_t)]\varphi_{\rm p}^2(x) dx,\\
&
(e^{-tL_{\rm p}} g)(x)=\Ebb^x[g(z_t)].\non
\end{align}
In the relativistic case we have $z_t=b_t$ and in the non-relativistic case 
$z_t=B_t$. See \kak{bp} and \kak{Bp}. 
In this paper we consider a $P(\phi)_1$-process associated with $\hn$. 
By a joint use of the  ground state-transform of $\hp$ and the Wiener-Segal-It\^o isomorphism $\fff\cong L^2(Q,d{\rm G})$  (see Section 2),
the components of $\hn $ can be mapped unitarily in a suitable space of functions. 
$\hn$ is transformed to the self-adjoint operator 
$$H=L_{\rm p} \otimes\one +\one \otimes \tilde{\hf}+\tilde\hi$$
in $L^2(\BR,\grp^2 dx)\otimes L^2(Q,d{\rm G})$. To simplify the notations, we write again $\hf$ for $\tilde{\hf}$, and $\hi$ for $\tilde\hi$.
The free field component ${\hf}$ gives rise to a $Q$-valued 
Ornstein-Uhlenbeck process $(\xi_t)_{t\in\RR}$ on a probability space 
$(\Y,\B(\Y),\G)$ (see \kak{10}) and 
it follows that 
$$(\Phi, e^{-t{\hf}}\Psi)=\Ebb_{\G}[\Phi(\xi_0)\Psi(\xi_t)].$$
The  particle  operator $L_{\rm p}=\grp^{-1}(\hp-E_{\rm p})\grp$
contributes by a $P(\phi)_1$-process $(z_t)_{t\geq0}$ which is a jump L\'evy-type process $(b_t)_{t\geq0}$ in the cases \eqref{rela}-\eqref{masslos} and a diffusion 
$(B_t)_{t\geq0}$ (Brownian motion) 
in the case \eqref{classi}, 
and it follows \eqref{FKF}. 
This allows to derive a Feynman-Kac type
representation of $e^{-tH}$  involving an infinite-dimensional joint particle-field random process $(z_t,\xi_t)_{t\geq0}$ generated by
\begin{equation}
\label{free}
H_0=L_{\rm p} \otimes\one +\one \otimes {\hf}.
\end{equation}

Since the semigroup operators $e^{-tH }$ are not Markovian, 
 the modified process becomes a Markov process only after a ground state transform is performed on $H $:
$$\LN =\frac{1}{\gr}(H-E)\gr.$$
See formula \eqref{nelgst} below. 
Our interest in this paper is to construct  $P(\phi)_1$-processes associated with  
$\LN$
and to derive the counterpart of \eqref{FKF}. 
The first main theorem in this paper is as follows. 

{\bf Main theorem 1.}\ 
{\it There exists a $P(\phi)_1$-process associated with $\LN $.} \\
The first main thorem is stated as Theorems \ref{relnel} for the relativistic case 
and Theorem \ref{classnel} for the non-relativistic case below. 

A major difference between the two cases appears in the path regularity (and thus support) properties of the corresponding 
$P(\phi)_1$-processes. 
While in the non-relativistic  case we have infinite-dimensional diffusions on a  sample space of functions 
allowing continuous paths, in the relativistic case we have more technical difficulties resulting from paths 
with jump discontinuities. However, since we are in the context of concrete models, we are able to control the jumps and path 
oscillations directly instead of using abstract arguments.

\subsection{Functional central limit theorem}
The classical central limit theorem says that 
the partial sums 
$$S_n=\frac{1}{\sqrt n} \sum_{i=1}^n X_i$$ of 
independently and identically distributed 
random variables $X_i$ with mean zero and covariance $\Ebb[X_i^2]=v$
weakly converge to a Gaussian random variable with mean zero and 
covariance $v$. 
We see that 
$$\frac{1}{n} \Ebb\left[\left|\sum_{i=1}^n X_i\right|^2\right]=v$$ and 
$\underset{n\to\infty} \lim \Ebb[g(S_n)]\to\Ebb[g(X)]$,  where $X$ is a Gaussian random variable 
with covariance $v$, and $g$ is a measurable function.  
On the other hand a functional central limit theorem (FCLT) considers random process 
$(X_t)_{t\geq0}$  and the limit of 
\begin{align}
\label{KV}
\frac{1}{\sqrt n}\int_0^{nt} (L f)(X_s) ds
\end{align} with some function $f$.
Bhattacharya   \cite[Theorem  2.1]{bha82}  shows that 
for a stationary and ergodic Markov process $(X_t)_{t\geq0}$ with self-adjoint generator $-L$ it follows that 
\begin{align}
\label{Bh}
\underset {n \to \infty}{\lim}\frac{1}{\sqrt {n}}\int_0^{nt} (L f)(X_s) ds
 \stackrel{\rm d} = 2 \lk f, Lf \rk B_t.
 \end{align}
 Here $A\stackrel{\rm d} =B$ implies that 
 random variables $A$ and $B$ have the same distribution. 
 See Section 4 for the detail.

A FCLT for the non-relativistic  Nelson model has been first established by Betz and Spohn in a different approach \cite{BS}. They
have shown that under the Gibbs measure obtained from taking the marginal over the particle-generated component of the path measure,
the appropriately scaled process converges in distribution to Brownian motion having reduced diffusion coefficients. 
This means that the particle increases its effective mass due to the coupling to the boson field. 
A main observation in their paper is that one can
associate a martingale with functionals \eqref{KV} of the process, whose long time behaviour can be predicted by using a martingale convergence
theorem. This result is originally due to Kipnis and Varadhan \cite[Theorem 1.8]{kv86} for more general Markov processes, see also the
ground-breaking work \cite{Don}, and related problems studied in \cite {bha82,CCG,ledoux11}. For the relativistic Nelson model the
problem has not been studied so far. In our approach we use the $P(\phi)_1$-process to derive FCLT behaviours rather than to construct
Gibbs measures obtained after integrating out the quantum field variables. Although this involves working with effectively infinite
dimensional processes, it is also more robust, and applicable to a range of other quantum field or solid state physics models.

An offshoot of constructions of $P(\phi)_1$-processes is that it allows to derive a FCLT  for additive functionals of the processes, relating
to particle-field operators in the relativistic and non-relativistic Nelson models. 
We emphasize that such a behaviour results also for the relativistic case. 
In Section 3 we construct $P(\phi)_1$-process $\pro{\tilde X}$ 
 whose self-adjoint generator is 
$-\LN $ with relativistic $\hp$ and 
we are interested in studying FCLT for $\pro {\tilde X}$. 
We consider 
$$\FF{t}=\int_0^t f(\tilde X_s) ds$$
where $f\in D(\LN^{-1/2})$. 
Here $\pro {\tilde X}$ will be seen as a stationary and ergodic Markov process with invariant  initial measure $\rm M$. 
Then we arrive at the second main theorem:

{\bf Main theorem 2.}\ 
{\it It follows that 
$$\lim_{s\to\infty} \frac{1}{\sqrt{s}}\FF{st}=\SS(f) B_t$$
in the sense of the weak convergence, where 
$\SS(f)=2 \|\LN^{-1/2} f\|^2_{\LM}$ and 
$\proo B$ is Brownian motion.}\\
The second main theorem is stated in Theorem \ref{main5} for both relativistic and non-relativistic cases. 

\subsection{Organization of the paper}
In Section 2 we briefly discuss the functional integral representations of the
non-relativistic  and relativistic Nelson models. In  Section 3 we construct $P(\phi)_{1}$-processes associated with the two models, and
present detailed proofs of path regularity and support properties. Section 4 is devoted to proving a FCLT 
for additive functionals associated with the Nelson models by using the properties of the $P(\phi)_{1}$-processes. We show some
functionals of special interest for both cases and determine explicitly the variance in the related FCLT. 

\section{Functional integral representations of the Nelson Hamiltonians}

\subsection{Functional integral representation of the quantum mechanical particle  Hamiltonians}
Let $V: \BR \to \RR$ be a Borel-measurable function giving the potential. We denote the multiplication operator defined by $V$
by the same label. The Schr\"odinger operators \eqref{rela}-\eqref{classi} can be defined in the sense of perturbation theory
by choosing suitable conditions on $V$. However, since in this paper we will use methods of path integration, a natural choice
is Kato-class, in each case given in terms of the related random processes.

Let $\RR^+=[0,\infty)$ and $d\geq 1$. 
We will consider the following processes and path spaces.\\
(Cauchy and relativistic Cauchy process)
Let $$\D=D(\RR^+,\BR)$$ be the space of c\`adl\`ag paths (i.e., continuous from right with left limits) on $\RR^+$ with values in $\BR$,
and $\B(\D)$ the $\s$-field generated by the cylinder sets of $\D$. 
The operator $-\hp^0$ in \eqref{rela} is the infinitesimal
generator of a rotationally symmetric relativistic Cauchy process, and in \eqref{masslos} the infinitesimal generator of a
rotationally symmetric Cauchy (i.e., $1$-stable) process. 
We will denote by $\proo b$, the coordinate process defined on 
the probability space $(\D,\B(\D), \ppp ^x)$ with path measure
$\ppp  ^x$ starting from $x$ at time $t=0$, ($\ppp$ when $x=0$). Namely $b_t(\omega)=\omega(t)$ for $\omega \in \D$. 
We note that although both of these processes have paths with jump
discontinuities, they have important qualitative differences resulting from the fact that the first has a jump measure which
decays exponentially, while the second decays polynomially at infinity.\\
(Brownian motion)
Let 
$$\X=C(\RR^+,\BR)$$ be the space of $\BR$-valued continuous functions on $\RR^+$. Also, let $\proo B$ be $d$-dimensional
Brownian motion defined on the probability space 
$(\X,\B(\X),\W^x)$, where $\B(\X)$ is the $\s$-field generated by the cylinder sets of $\X$, and
$\W^x$ Wiener measure starting from $x$ at $t=0$, ($\W$ when $x=0$). 
Its infinitesimal generator is $-\hp^0$ as in \eqref{classi}.

It is well-known that all three processes have the strong Markov property with respect to their natural filtrations.
By using these processes we can define the related spaces for the potentials.
\begin{definition}
\label{kato}
{\rm
Let $\proo Z$ be a random process on a probability space  
$(X,{\mathcal  B}(X),m^x)$. 
We say that $V = V_+ - V_-$ is a {Kato-class potential} with respect to $\proo Z$ whenever for
its positive and negative parts
$
V_- \in {\mathcal  K}^Z $ and $V_+  \one_C \in {\mathcal  K}^Z$ 
{for every compact set $C \subset \RR^d$}
hold, where $f \in {\mathcal  K}^Z$ means that
$$
\lim_{t \downarrow 0} \sup_{x \in \RR^d} \Ebb_m^x \left[\int_0^t |f(Z_s)| ds\right] = 0.$$ 
When $\proo Z = \proo B$, we call this space {Kato-class}, and when $\proo Z = \proo b$, we call it {relativistic
Kato-class}.
}
\end{definition}

By Khasminskii's lemma \cite[Lemma 3.37]{LHB} and its straightforward extension to relativistic Kato-class it follows
that the random variables $-\int_0^t V(b_s) ds$ are exponentially integrable for all $t \geq 0$, and thus we can define
the Feynman-Kac semigroup
$$
T_t f(x) = \Ebb_{\ppp}^x\left[e^{-\int_0^t V(b_s) ds} f(b_t)\right]$$ 
Using the Markov property and stochastic continuity of $\proo b$ it can be shown that $\{T_t: t\geq 0\}$ is a strongly
continuous symmetric one-parameter semigroup 
on $L^2(\RR^d)$. The same considerations also hold when the
process is Brownian motion. For all three processes respectively, by the semigroup version of Stone's theorem \cite[Proposition 3.26]{LHB}  
in each case there
exists a self-adjoint operator $K$ bounded from below such that $e^{-tK} = T_t$. Using the generator $K$, we can give a
definition to the operators \eqref{rela}-\eqref{classi} for Kato-class potentials, however, we keep using the notations
$\hp$ for simplicity. In both the non-relativistic and relativistic cases we have then the following Feynman-Kac formulae.

\begin{description}
\item[(Relativistic case)]
Let $\hp$ be given by \eqref{rela} or \eqref{masslos}, and consider the related random process $\proo b$. If $V$ is of
relativistic Kato-class, then
\begin{align}\label{bp}
(f, e^{-t\hp}g)_\LR=\int_{\BR}\Ebb^x_\ppp [\bar{f}(b_0)e^{-\int_0^tV(b_s)ds} g(b_t)] dx.
\end{align}
\item[(Non-relativistic case)]
Let $\hp$ be given by \eqref{classi} and consider Brownian motion $\proo B$. If $V$ 
is of  Kato-class, then
\begin{align}\label{Bp}
(f, e^{-t\hp}g)_\LR=\int_{\BR}\Ebb^x_\W[\bar{f}(B_0)e^{-\int_0^tV(B_s)ds} g(B_t)] dx.
\end{align}
\end{description}
The proof for the relativistic case see \cite[Section 4]{HIL12}, and for the non-relativistic  case we refer to \cite[Sections 3.3, 3.6]{LHB}.
From this  we can derive 
a functional integral representation of 
$(e^{-t\hp}g)(z)$  as 
\begin{align}\label{V}
(e^{-t\hp}g)(x)=\Ebb^x_m[e^{-\int_0^tV(X_s)ds} g(X_t)],
\end{align}
where $X_t=b_t$ and $m^x=\tilde \W^x$ 
for the relativistic case and $X_t=B_t$ and $m^x=\W^x$ for the  non-relativistic case. 

\subsection{$P(\phi)_1$-process associated with $\hp$}
Let 
$$\DD=D(\RR,\BR)$$
be the space of paths on $\RR$ with values in $\BR$,
and $\B(\DD)$ the $\s$-field generated by the cylinder sets of $\DD$. 
Let 
$$\XX=C(\RR,\BR)$$ 
be the space of $\BR$-valued continuous functions on $\RR$
and $\B(\XX)$ the $\s$-field generated by the cylinder sets of $\DD$. 
We will need two-sided processes $\pro b$ and $\pro B$ on 
$(\DD,{\mathcal  B}(\DD),\ppp^x)$ and $(\XX,{\mathcal  B}(\XX),{\W}^x)$, respectively, i.e., indexed by the time-line $\RR$ instead of usually by the semi-axis $\RR^+$. 
The properties
of the so obtained process can be summarized as follows.
See \cite{LO12} for details.
The following hold:
\begin{description}
\item[(i)]
${\ppp  }^x(b_0=x)=1$, $x \in \BR$,
\item[(ii)]
the increments $(b_{t_i}-b_{t_{i-1}})_{1\leq i\leq n}$ are independent symmetric relativistic Cauchy (resp. Cauchy)
random variables for any $0=t_0<t_1<\ldots<t_n$ with $b_t-b_s\stackrel{\rm d}{=}b_{t-s}$ for
$t>s>0$,
\item[(iii)]
the increments $(b_{-t_{i-1}}-b_{-t_i})_{1\leq i\leq n}$ are independent symmetric relativistic Cauchy (resp. Cauchy)
random variables for any $0=-t_0>-t_1>\ldots>-t_n$ with $b_{-t}- b_{-s}\stackrel{\rm d}{=}b_{s-t}$
for $-t>-s>0$,
\item[(iv)]
the function $\RR \ni t\mapsto b_t(\omega)\in \BR$ is c\`adl\`ag for every 
$\omega \in \DD$,
\item[(v)]
the random variables $b_t$ and $b_s$ are independent for $t>0$ and $s<0$.
\end{description}
A completely similar construction can be done to obtain two-sided Brownian motion $\pro B$, with simplifications due to
path continuity.

We give a remark on notations about measures $\ppp^x$ and $\W^x$.
In what follows notation $\ppp^x$ (resp. $\W^x$) 
is used as a probability measure on both measurable spaces 
$(\D, {\mathcal B}(\D))$ and 
$(\DD, {\mathcal B}(\DD))$ (resp. 
$(\X, {\mathcal B}(\X))$ and 
$(\XX, {\mathcal B}(\XX))$) whenever confusion may not arise.

\begin{definition}\TTT{$P(\phi)_1$-process}
\label{matu}
\rm{
Let $(E, \ms F, P)$ be a probability space and $K$ be a self-adjoint operator in 
$L^2(E, dP)$,  bounded from below.
We say that an $E$-valued random  process $\pro Z$ 
on a probability space $({\mathcal  Y}, \B_{\mathcal  Y} , \Q^z)$ is a
{$P(\phi)_1$-process} associated with $((E, \ms F, P), K)$ if conditions (i)-(iv) below are satisfied:
\begin{description}
\item [(i)] (Initial distribution)
$\Q^z(Z_0=z)=1$.
\item[(ii)] ({Reflection symmetry})
$(Z_t)_{t\geq 0}$ and $(Z_t)_{t\leq0}$ are independent and ${Z}_t\stackrel{\rm d}{=}{Z}_{-t}$, $t\in\RR$.
\item [(iii)] ({Markov property})
$({Z}_t)_{t\geq 0}$ and $({Z}_t)_{t\leq0}$ are Markov processes with respect to 
 $\sigma(Z_s,0\leq s\leq t)$ and 
$\sigma(Z_s,t\leq s\leq 0)$, respectively. Here $\sigma(Z_s,0\leq s\leq t)$ (resp. $\sigma(Z_s,t\leq s\leq 0)$) denotes the minimal $\sigma$-field generated by 
$\{Z_s\mid 0\leq s\leq t\}$ (resp. $\{Z_s\mid t\leq s\leq 0\}$).   
\item [(iv)]  ({Shift invariance and functional integration})
Let $-\infty<t_0\leq t_1<\ldots\leq t_n<\infty$, $f_j\in L^{\infty}(E,dP)$, $j=1,\ldots,n-1$, and $f_0,f_n\in L^2(E,dP)$. 
Then for every $s\in\RR$,
\begin{align}
\hspace{-0.5cm}
\int_{E}\Ebb_{\Q^z}\!\!\!
\left[ \prod_{j=0}^{n}f_{j}({Z}_{t_{j}})\right]  dP(z) =
\int_{E}\Ebb_{\Q^z}\!\!\!
\left[ \prod_{j=0}^{n}f_{j}({Z}_{t_{j}+s})\right]  dP(z)
=(\bar f_0, e^{-(t_{1}-t_{0})K }f_1\cdots
f_{n-1}e^{-(t_{n}-t_{n-1})K}f_n).
\end{align}
In particular it follows that for $f,g\in L^2(E,dP)$
\begin{align*}
&(f, e^{-tK }g)_{L^2(E)}=
\int_{E}\Ebb_{\Q^z}\left[\bar f(Z_0)g(Z_t)\right]  dP(z),\\
&
 (e^{-tK }g)(z)=
\Ebb_{\Q^z}\left[g(Z_t)\right].
\end{align*}
\end{description}
}
\end{definition}
 
Denote
\begin{equation}
\label{N}
d{\rm N}(y)=\grp^2(y) dy, \quad y \in \BR.
\end{equation}
Since by part (iv) of Assumption \ref{assumption} the function $\grp$ is square integrable and $L^2$-normalized,
$d{\rm N}$ is a probability measure on $\BR$. Define the unitary operator
\begin{equation}
\label{gstransf}
U_{\rm p}: L^{2}(\BR ,d{\rm N}) \rightarrow L^{2}(\BR,dy), \quad f \mapsto  \grp f.
\end{equation}
Recall that $E_{\rm p} = \inf \Spec(\hp)$. Using that $\grp$ is strictly positive by the same assumption, 
\begin{align}
\label{Lp}
L_{\rm p}=U_{\rm p}^{-1}(\hp -E_{\rm p})U_{\rm p},
\end{align}
is well-defined and  $D(\lp)=
\{f\in L^2(\BR,d{\rm N}) \mid  f\gr\in D(\hp)\}$. Since $e^{-t\lp} \one=\one$ for the identity function $\one\in
L^2(\BR,d{\rm N})$, the operator $\lp$ is the infinitesimal generator of a Markov process.

\medskip
\begin{proposition}
[$P(\phi)_1$-process associated with $\hp$]
\label{pphi}
\begin{description}
\item[\rm (Relativistic case)]
Let $\hp$ be the relativistic Schr\"odinger operator \eqref{rela} or \eqref{masslos}, with $V$ in relativistic
Kato-class, and consider $L_{\rm p}$ defined in \eqref{Lp}. Then
\begin{description}
\item[\rm (i)]
there exists a probability measure $\tilde {\N}^y$ on 
the measurable space $(\DD, \B(\DD))$ such that the coordinate process
$\pro {\tilde Y}$ on the probability space $(\DD, \B(\DD),\tilde {\N}^y)$ is a $P(\phi)_1$-process starting from $y \in \BR$, associated with
$((\BR, \B (\BR), {\rm N}),{L}_{\rm p})$,
\item[\rm (ii)]
the function $t\mapsto \tilde Y_t$ is almost surely c\`adl\`ag.
\end{description}

\item[\rm (Non-relativistic case)]
Let $\hp$ be the non-relativistic  Schr\"odinger operator \eqref{classi}, with $V$ in Kato-class, and consider $L_{\rm p}$
defined in \eqref{Lp}. Then
\begin{description}
\item[\rm (i)]
there exists a probability measure ${\N}^y$ on the measurable space $(\XX, \B(\XX))$ such that the coordinate process $\pro{Y}$
on the probability space $(\XX, \B(\XX),{\N}^y)$ is a $P(\phi)_1$-process starting from $y \in \BR$, associated with 
$((\BR,\B (\BR), {\rm N}), \lp)$, 
\item[\rm (ii)]
the function $t\mapsto {Y}_t$ is almost surely continuous.
\end{description}
\end{description}
\end{proposition}
\proof
We refer to \cite[Theorem 5.1]{LO12} for relativistic case,   
and 
\cite[Theorem 3.106]{LHB} for non-relativistic case.
\qed

Proposition \ref{pphi} yields that 
\begin{align}
&(f, e^{-t\lp }g)_{L^2(\BR,d{\rm N})}=
\int_{\BR}\Ebb_{M^y}\left[\bar f(X_0)g(X_t)\right]  d{\rm N}(y),\\
&\label{pFKF}
(e^{-t\lp }g)(y)
=
\Ebb_{M^y}\left[g(X_t)\right],
\end{align}
where $M=
\tilde \N$ and $X_t=
\tilde Y_t$ for the relativistic case and 
$M=
\N$ and $X_t=
Y_t$ for the non-relativistic case. 
We have functional integrals for $\hp$ and $\lp$ in \eqref{V} and \eqref{pFKF}, respectively.   From these presentations formally we can see that 
there is no Markov process 
which has generator $\hp$ but 
$\lp$ 
is the generator of $P(\phi)_1$-process. 
From now on we drop the tilde-notations and it will be clear from the context if a statement applies
to any of these operators/processes or to a particular operator/process.

\subsection{Infinite dimensional Ornstein-Uhlenbeck process}
Recall $\omega(k)=\sqrt{|k|^2+\nu^2}$. 
Let $\K $ be a Hilbert space over $\RR $, defined by the completion of $D(1/\sqrt\omega)$ with respect to the
norm determined by 
$
(f,g)_{\K }=(1/2)(\hat f/\sqrt\omega, \hat g/\sqrt\omega)_\LR$,
i.e.,
$
\K=\overline{D(1/\sqrt\omega)}^{\|\cdot\|_\K}$. 
Let $T:\K\to\K$ be a positive self-adjoint operator with a Hilbert-Schmidt inverse  
such that $\sqrt{\omega}{T}^{-1}$ is
bounded. Define 
$ \K _{n}=\overline{\cap_{p=1}^\infty D(T^p)}^{\|{T}^{n/2}\cdot \|_{\K}}$. 
Identifying 
$\K ^\ast _{+2}=\K _{-2}$, 
we construct a triplet $\K _{+2}\subset \K\subset\K _{-2}$.
 Write
$Q=\K _{-2}$, and endow $Q$ with its Borel $\s$-field $\B (Q)$.
Let $\xi(f)=\lk \!\lk \xi_{t}, f\rk \!\rk$
for $f\in \K _{+2}$, where $\lk \!\lk .,.\rk \!\rk $ denotes the pairing between $Q$ and 
$\K _{+2}$. 
We can construct a probability measure $\rm G$ on $(Q, \B(Q))$ such that 
\begin{align}\label{eq2}
\Ebb_{\G }[\xi(f)\xi(g)]=(f,g)_{\K }.
\end{align}
Note that by \kak{eq2} every $\xi(f)$ can be uniquely extended to test functions 
$f\in\K $, which for simplicity we will
denote in the same way.
Thus we are led to the probability space $(Q, \B(Q),{\rm G})$. 
 We shall see that $\fff$ is unitary equivalent to $L^2(Q,d\rm G)$ in Section \ref{bus} below. 
Consider the set 
$$\Y =C(\RR , Q)$$ of continuous functions on $\RR$, with values in $Q$, and denote its Borel $\s$-field
by $\mathcal B(\Y)$. 
We can define a {$Q$-valued Ornstein-Uhlenbeck  process} $\pro \xi$,
$
\RR \ni t\mapsto\xi_{t}\in Q,
$
on the probability space 
\begin{align}\label{10}
(\Y ,\B (\Y), \G )
\end{align}
with a probability measure $\G $. 
Let $\xi_{t}(f)=\lk \!\lk \xi_{t}, f\rk \!\rk$
for $f\in \K _{+2}$. 
For every $t \in
\RR$ and $f$ we have that $\xi_{t}(f)$ is a Gaussian random variable with mean zero and covariance
\begin{align}\label{eq1}
\Ebb_{\G }[\xi_{t}(f)\xi_{s}(g)]=(f,e^{-|t-s|\omega}g)_{\K}.
\end{align}
Note that by \kak{eq1} every $\xi_{t}(f)$ can be uniquely extended to test functions $f\in\K $, which for simplicity we will
denote in the same way.
The notation $\G ^{\xi}(\cdot)=\G (\cdot  | \xi_{0}=\xi)$  is a regular conditional probability measure. 

\medskip
\subsection{Functional integral representation for Nelson's Hamiltonian}
\label{bus}
To obtain a path integral representation of the Nelson Hamiltonian, first we need a representation of $\hn $ in
a suitable function space. The quantum mechanical particle  operator $\hp$ (any of the three above) acting on $L^2(\BR)$ can be treated
by a ground state transform as in \eqref{Lp}. To transform the boson field operators, we use the Wiener-It\^{o}-Segal
isomorphism
$$
\WIS :\fff \rightarrow L^{2}(Q, d{{\rm G}}), \quad \WIS  \phi(f)\WIS ^{-1}=\xi(f).
$$
Here $\phi(f)$ is the field operator defined in Remark \ref{NN}. 
In combination with
\eqref{N}, define the product measure
\eq{pzero}
{\rm P}={\rm N}\otimes {{\rm G}},
\en
which is a probability measure on the product space 
$$(\hir,\B_\hir)=(\BR\times Q,\B(\BR)\otimes\B(Q)).$$ 
The map
$
U_{\rm p}\otimes\WIS :\hhh \rightarrow L^{2}(\hir , d{\rm P})
$
establishes a unitary equivalence between $L^{2}(\hir , d{\rm P})$ and $\hhh $, and we make the identification
\eq{identification}
\hhh \cong 
L^{2}(\hir,d{\rm P}).
\en
For convenience, hereafter we write simply as 
$$L^{2}(\hir , d{\rm P})=L^{2}({\rm P}),\quad 
L^{2}(\BR, d{\rm N})=L^{2}({\rm N}),\quad
L^{2}(Q, d{{\rm G}})=L^{2}( {{\rm G}}).$$ 
Thus we have $L^{2}({\rm P})\cong 
L^{2}({\rm N})\otimes L^{2}( {{\rm G}})$.  
The images of the free field
and interaction Hamiltonians on $L^2({\rm P})$ 
under this unitary map are given by
$\hf \mapsto \WIS  \hf \WIS ^{-1}$ and 
$\hi (y) \mapsto \WIS  \hi \WIS ^{-1}(y)=\xi(\tilde{\varphi}(\cdot-y))$,
where $\tilde{\varphi}$ is the inverse Fourier transform of $\vp/\sqrt\omega$. 
Let $\tilde{\hf}=\WIS  \hf \WIS ^{-1}$. 
We have 
$$(\Phi, e^{-t\tilde{\hf}}\Psi)_{L^2({\rm G})}=
\Ebb_{\G}[\Phi(\xi_0)\Psi(\xi_t)]
=\int_Q \Ebb_{\G^\xi}[\Phi(\xi_0)\Psi(\xi_t)]d{\rm G}(\xi).$$

To simplify the notations, we write again
$\hf$ and $\hi$ for the above images. 
The Nelson Hamiltonian $\hn $ is unitary equivalent with
\begin{align}
H= L_{\rm p}\otimes\one +\one \otimes \hf +\hi
\end{align}
acting on $L^{2}({\rm P})$, where $L_{\rm p}$ is given by (\ref{Lp}).

Using the $P(\phi)_{1}$-process $\pro{\tilde Y}$ (resp. $\pro Y$) on 
$(\DD ,\B(\DD), \tilde \N^y)$ (resp. $(\XX ,\B(\XX), \N^y)$) 
associated with 
$((\BR,\B(\BR),{\rm N}), \lp)$ as given in the relativistic case (resp. non-relativistic case) in 
Proposition \ref{pphi}, we write
$$
d\tilde \N=d{\rm N}(y)d\tilde \N^y
\quad (resp. \ \ d \N=d{\rm N}(y)d \N^y).
$$
The probability space for the joint system in the relativistic case (resp. non-relativistic case) 
is then the product space 
$$(\DD\times\Y, \B(\DD)\otimes\B(\Y), \tilde \P_0)\quad (resp.\ \ 
(\XX\times\Y, \B(\XX)\otimes\B(\Y), \P_0)),$$
where 
$$
\tilde \P _{0}=\tilde \N \otimes\G\quad (resp.\ \ 
\P _{0}= \N \otimes\G).
$$
Define the shift operator $\tau_s: L^2({\rm G})\mapsto L^2({\rm G})$ by $\tau_s\xi(h)=\xi(h(\cdot-s))$. We have
then the following functional integral representation for the Nelson Hamiltonian $H$ in $L^2({\rm P})$.
\begin{proposition}
\label{fik}
Let $\Phi,\Psi\in L^{2}({\rm P})$ and  suppose that $s\leq 0\leq  t$. Then
\begin{align*}
&{\rm (relativistic\  case)}\ \ 
(\Phi,e^{-(t-s)H}\Psi)_{L^{2}({\rm P})}= \Ebb_{\tilde \P _{0}}[\Phi(\tilde Y_{s},\xi_{s})e^{-\int_{s}^{t}
\tau_{\tilde Y_{r}}\xi_{r}(\vpp)dr}\Psi(\tilde Y_{t},\xi_{t})],\\
&{\rm (non\!-\!relativistic\  case)}\ \ 
(\Phi,e^{-(t-s)H}\Psi)_{L^{2}({\rm P})}= \Ebb_{\P _{0}}[\Phi(Y_{s},\xi_{s})e^{-\int_{s}^{t}
\tau_{Y_{r}}\xi_{r}(\vpp)dr}\Psi(Y_{t},\xi_{t})].
\end{align*}
\end{proposition}
\proof
For the non-relativistic  Nelson model see \cite[Theorem 6.2]{LHB}. For the relativistic case the proof can be made similarly.
\qed

\subsection{Positivity improving}

For later use we have the following representation formula using 
the specific process
instead of  a $P(\phi)_1$-process.
\begin{proposition}
\label{B}
Let $\Phi,\Psi\in L^{2}({\rm P})$ and $s\leq 0\leq  t$. Then we have for the relativistic cases that
\begin{align*}
(\Phi,e^{-(t-s)H}\Psi)_{L^{2}({\rm P})}= \int_{\hir} \Ebb_{\ppp \otimes{\G}} ^{(y,\xi)}\left[\Phi(b_s,\xi_s)\grp(b_s)
e^{-\int_s^t\tau_{b_{r}}\xi_r(\vpp)dr} \Psi(b_t,\xi_t)\grp(b_t) e^{-\int_s^t V(b_r) dr}\right] dy\otimes d{{\rm G}}.
\end{align*}
The same relation holds for the non-relativistic  case on replacing $\ppp $ by $\W$ and $\pro b$ by $\pro B$,i.e.,
\begin{align*}
(\Phi,e^{-(t-s)H}\Psi)_{L^{2}({\rm P})}= \int_{\hir} \Ebb_{\W \otimes{\G}} ^{(y,\xi)}\left[\Phi(B_s,\xi_s)\grp(B_s)
e^{-\int_s^t\tau_{B_{r}}\xi_r(\vpp)dr} \Psi(B_t,\xi_t)\grp(B_t) e^{-\int_s^t V(B_r) dr}\right] dy\otimes d{{\rm G}}.
\end{align*}
\end{proposition}
\proof
For the non-relativistic  Nelson model see \cite[Theorem 6.3]{LHB}. For the relativistic case a similar proof can be applied,
the details are left to the reader.
\qed

We give the definition of 
positivity preserving operators and positivity improving operators below.
\begin{definition}
Let $(E,\nu)$ be a $\sigma$-finite measure space. 
Let $S$ be a bounded  operator on $L^2(E,d\nu)$. 
\bi
\item[(i)]
$S$ is positivity preserving if and only if 
$(f, Sg)_{L^2(E)}\geq 0$ for any $f\geq$ and $g\geq0$.
\item[(ii)]
$S$ is positivity improving if and only if 
$(f, Sg)_{L^2(E)}>0$ for any $f\geq$ and $g\geq0$ but $f\not\equiv 0$ and $g\not\equiv 0$.
\ei
\end{definition}
Form the definition we can see that when $S$ is positivity improving, 
$(Sf)(x)>0$ for a.e. $x\in E$  for any non-negative function $f\geq0$ $(f\not\equiv 0$). 
Furthermore in the case of $S=e^{-tK}$ with some self-adjoint operator bounded  below, 
we have a useful proposition below. 
\begin{proposition}
\label{den5}
Let $(E,\nu)$ be a $\sigma$-finite measure space. 
Let $K$ be a self-adjoint operator that is bounded below on $L^2(E,d\nu)$. 
Suppose that $e^{-tK}$ is positivity preserving for all $t>0$ and $\is(K)$ is an eigenvalue.
Then the following two statements are equivalent 
\bi
\item[\rm (i)] $\is(K)$ is a simple eigenvalue with a strictly positive eigenfunction. 
\item[\rm (ii)] $e^{-tL}$ is positivity improving for all $t>0$. 
\ei
\end{proposition}
\proof
See e.g., \cite[XIII.12]{rs4}.  
\qed
\begin{remark}\label{PI}
By (iv) of  Assumption \ref{assumption} and Proposition \ref{den5} 
it is seen that $e^{-t\hn}$ is positivity improving. 
We also give a sufficient condition for the positivity improving property
of $e^{-t\hn}$ in Lemma \ref{unique} below. 
\end{remark}

\begin{lemma}
\label{unique}
Let $X_t=b_t$ and $m=\tilde {\W}$ 
for the relativistic case and $X_t=B_t$ and $m=\W$ for the  non-relativistic case. 
Suppose that there exists $t>0$ such that
$$
\int_\BR m^x\lk  e^{-\int_0^t V(X_s) ds}=0\rk  dx =0.$$
Then $e^{-t\hn}$ is positivity improving. 
In particular, for $V_+\in L_{\rm loc}^1(\BR)$, if $\hn$ has a
ground state, then it is unique.
\end{lemma}
\proof
It is sufficient to show that 
$(\Phi, e^{-t\hn}\Psi)>0$ for any $\Phi\geq0$ and $\Psi\geq0$ but $\Phi\not\equiv0\not\equiv \Psi$. 
By \cite[Theorem 3.55]{LHB}, $e^{-t\hp}$ is positivity improving. 
It is also established that $e^{-t\hf}$ is  positivity improving as well. 
Together with them it can be seen that 
$e^{-t(\hp^0\otimes\one+\one\otimes\hf)}$ is also 
positivity improving. We have 
\begin{align*}
(\Phi,e^{-t\hn}\Psi)_{L^{2}({\rm P})}
= \int_{\hir} \Ebb_{m \otimes{\G}} ^{(y,\xi)}\left[\Phi(X_0,\xi_0)
e^{-\int_0^t\tau_{X_{r}}\xi_r(\vpp)dr} \Psi(X_t,\xi_t) e^{-\int_0^t V(X_r) dr}\right] dy\otimes d{{\rm G}}\geq0.
\end{align*}
Suppose that $(\Phi,e^{-t\hn}\Psi)_{L^{2}({\rm P})}=0$. 
We see that 
$\| \int_0^t\tau_{X_{r}}\xi_r(\vpp)dr\|_{L^2(Q)}\leq t
(\|\vp/\omega\|+\|\vp/\sqrt\omega\|)$ and 
the measure of the sets  
\begin{align*}
&{\rm (relativistic\ case)} \quad \left\{(y,\omega,\xi ) \in \BR\times \tilde {\ms X}\times Q 
\left|  |
\int_0^t\tau_{y+X_{r}(\omega)}\xi_r(\vpp)dr |=\infty\right.\right\},\\
&{\rm (non\!-\!relativistic\ case)} \quad \left\{(y,\omega,\xi ) \in \BR\times {\ms X}\times Q \left|  
|
\int_0^t\tau_{y+X_{r}(\omega)}\xi_r(\vpp)dr |=\infty\right.\right\}
\end{align*}
are  zero. 
Hence  $(\Phi,e^{-t\hn}\Psi)_{L^{2}({\rm P})}=0$
implies that 
$$\int_{\hir} \Ebb_{m \otimes{\G}} ^{(y,\xi)}\left[\Phi(X_0,\xi_0)
 \Psi(X_t,\xi_t) \right] dy\otimes d{{\rm G}}=0$$ 
 by the assumption,
but this contradicts that 
$$0<(\Phi, 
e^{-t(\hp^0\otimes\one+\one\otimes\hf)}\Psi)= \int_{\hir} \Ebb_{m \otimes{\G}} ^{(y,\xi)}\left[\Phi(X_0,\xi_0)
 \Psi(X_t,\xi_t) \right] dy\otimes d{{\rm G}}.$$ 
We conclude that $(\Phi,e^{-t\hn}\Psi)_{L^{2}({\rm P})}
>0$ and the lemma follows. 
\qed

\section{$P(\phi)_{1}$-processes associated with the Nelson Hamiltonians}\label{sec}
\subsection{$P(\phi)_{1}$-process associated with the relativistic Nelson model}
In this section we construct $P(\phi)_1$-processes associated with the relativistic and non-relativistic  Nelson models. The two
cases differ significantly by the sample path regularity (and thus also support) properties, requiring different
treatment. First we consider the relativistic Nelson model for the particle operators \eqref{rela} and \eqref{masslos},
thus the process $\pro b$ below will mean either a relativistic Cauchy or a Cauchy process.

Since the ground state $\gr$ of $H$ is strictly positive and $L^2$-normalized, we can define the probability
measure
\begin{equation}
\label{rmm}
d{\rm M}=\gr  ^{2}d {\rm P},
\end{equation}
on $\hir$. Also, consider the unitary operator
$$
{U}_{\rm g} :L^{2}(\hir, d{\rm M})\rightarrow L^{2}(\hir , d {\rm P}), \quad \Phi \mapsto  \gr  \Phi,
$$
and denote 
$$\LM =L^{2}(\hir, d{\rm M}).$$
Recall that $E = \inf \Spec({H}_{\rm N})$.  
We define the ground
state-transformed  operator by
\begin{equation}
\label{nelgst}
{L}_{\rm N}={U}_{\rm g}^{-1}( H - E){U}_{\rm g}
\end{equation}
with the domain 
$$D(\LN)=\{f\in \LM  \mid f\gr\in D(H)\}=\{g/\gr \mid g\in D(H)\}.$$
It is immediate to see the lemma below.
\begin{lemma}\label{denden}
 $\LN$ is self-adjoint and 
  $e^{-t\LN}$ is positivity improving for $t>0$.
In particular $\LN f$ is real when $f\in D(\LN)$ is real. 
\end{lemma}
\proof
Since $\gr>0$ we see that 
 ${\rm Ker}(\LN\pm i\one)={\rm Ker}(\gr^{-1}(H-E+\pm i\one)\gr)=
{\rm Ker}(H-E+\pm i\one)=\{0\}$. This implies $\LN$ is self-adjoint.
Let $\Psi,\Phi\in \LM $ and 
$\Psi\geq0$ and $\Phi\geq0$ but $\Phi\not\equiv0\not\equiv \Psi$.
We see that 
$$(\Psi, e^{-t\LN}\Phi)_{ \LM }=
(\gr \Psi, e^{-t(H-E)}\gr \Phi)_{ L^2({\rm P})}>0$$ since 
$\gr \Psi\geq0$ and $\gr \Phi\geq0$. Thus 
  $e^{-t\LN}$ is positivity improving for $t>0$. 
Let $f\in D(\LN)$ is real. 
  We see that $e^{-t\LN} f $ is also real. Since $\LN f$ equals  the derivative of $e^{-t\LN} f$ at $t=0$, 
  $\LN f$ is also real. 
  \qed

Let $$\DD_\hir =D(\RR,\hir)$$ be the space of c\`adl\`ag paths indexed by $\RR$, with values in $\hir$, and $ \tilde \C_\hir$
the $\sigma$-field generated by cylinder sets. The following is our first main result.
\bt{pp11}
\TTT{$P(\phi)_1$-process associated with the relativistic Nelson Hamiltonian}
\label{relnel}
Let $(y,\xi)\in \hir$. The following hold.
\begin{description}
\item[\rm (i)]
There exists  a probability measure $\tilde \P^{(y,\xi)}$ on the measurable space $( \DD_\hir , \tilde \C_\hir)$ such that the $\hir$-valued 
coordinate process $\pro {\tilde X}$ on $( \DD_\hir , \tilde \C_\hir,\tilde \P^{(y,\xi)})$ is a $P(\phi)_1$-process associated with 
$((\hir, \B_\hir , {\rm M}), {L}_{\rm N}).$ In particular 
\begin{align}
&(\Psi, e^{-t \LN }\Phi)_{L^2(\rm P)}=\int_\hir \Ebb_{\tilde \P}^{(y,\xi)}[\Psi(\tilde X_0)\Phi(\tilde X_t)]d{\rm M},\\
&(e^{-t \LN }\Phi)(y,\xi)=\Ebb_{\tilde \P}^{(y,\xi)}[\Phi(\tilde X_t)].
\end{align}
\item[\rm (ii)]
The function $t\mapsto \tilde X_t$ has a c\`adl\`ag version.
\end{description}
\et

In order to  prove this theorem we need a string of lemmas below. 
Let $\B_\hir$ be the family of Borel  measurable sets of $\M$. 
Define the family of set functions $\{\M  _{\Lambda} \mid \Lambda\subset[0,\infty),   \#\Lambda <\infty\}$ on the $\#\Lambda$-fold product of $\B_\hir$: 
$\B_\hir ^{\#\Lambda}=\underbrace { \B_\hir  \times\cdots\times \B_\hir }_{\#\Lambda-\mbox{\tiny{times}}}$ 
by
\begin{align}
\M  _{\Lambda}(A_{0}\times A_{1}\times\cdots \times A_{n})
=\lk  \one _{A_{0}},
e^{-(t_{1}-t_{0})\LN }\one _{A_{1}}
\cdots\one _{A_{n-1}}e^{-(t_{n}-t_{n-1})\LN }\one _{A_{n}}  \rk_{\LM}
\label{mmeasure}
\end{align}
for $\Lambda=\{t_0,\ldots,t_n\}$. It is straightforward to show that the family of set functions
$\M_{\Lambda}$ satisfies the Kolmogorov consistency relation
$$
\M  _{\{t_{0},t_{1},\ldots,t_{n+m}\}}((\times_{i=0}^{n}A_{i})\times (\times_{i=n+1}^{n+m}\hir ))=
\M  _{\{t_{0},t_{1},\ldots,t_{n}\}} (\times_{i=0}^{n}A_{i}), \quad m, n \in \mathbb N.
$$
Define the projection 
$$\pi_{\Lambda}: \hir ^{[0,\infty)}\rightarrow  \hir ^{\Lambda}$$ by $w \mapsto
(w(t_{0}),\ldots,w(t_{n}))$ for $\Lambda=\{t_{0},\ldots,t_{n}\}$. 
Then
$
{\mathcal  A} = 
\{\pi_{\Lambda}\f(A) \mid  A\in \B_\hir ^{\#\Lambda},  \#\Lambda<\infty \}
$
is a finitely additive family of sets, and  the Kolmogorov extension theorem 
yields
that there exists a unique probability measure $\M$ on $(\hir ^{[0,\infty)},\sigma({\mathcal  A} ))$ such that
$
\M  (\pi_{\Lambda}^{-1}(A_{1}\times\cdots\times A_{n})) = \M  _{\Lambda}(A_{1}\times\cdots\times A_{n})$ 
for all $\Lambda\subset[0,\infty)$ with $\#\Lambda<\infty$ and $A_{j}\in \B_\hir $, and
\begin{align}
\label{p2}
\M _{\{t_0,\ldots,t_n\}}(A_0\times\cdots\times A_n)=\Ebb_\M \left[\prod_{j=0}^n \one_{A_{j}}(Z_{t_j})\right]
\end{align}
holds. Here $(Z_t)_{t\geq0}$ is an $\hir$-valued coordinate process defined by $Z_t(\omega)=\omega(t)$ for $\omega \in \hir ^{[0,\infty)}$.
From \eqref{p2} the following result is immediate.

\begin{lemma}\label{shift2}
The random  process $({Z}_t)_{t\geqslant0}$ is time-shift invariant under ${\M}$, i.e., for every $f_0,\ldots,f_n
\in {\K} $ and $s\geqslant0$ it follows that
\begin{align}\label{shi2}
\Ebb_{{\M}}\left[ \prod_{j=0}^{n}f_{j}({Z}_{t_{j}+s})\right] =\Ebb_{{\M}}
\left[ \prod_{j=0}^{n}f_{j}({Z}_{t_{j}})\right] =\lk  \bar f_{0},
e^{-(t_{1}-t_0)
{L}_{\rm N}}f_{1}\cdots f_{n-1}e^{-(t_n-t_{n-1}){L}_{\rm N}}f_{n}\rk _{{\K} }.
\end{align}
\end{lemma}

Now we prove that $({Z}_t)_{t\geqslant0}$ has a c\`adl\`ag version under ${\M}$. For this purpose, we need the
following technical lemmas which make use of the general ideas in \cite[pp 59-62]{Sato}. 
Let $I\subset[0,\infty)$
and $\varepsilon>0$. We say that ${Z}_\cdot(\omega)$, with fixed $\omega$, has $\varepsilon$-oscillation $n$ times
in $I$ if there exist  $t_0,t_1,\ldots,t_n\in I $ such that $t_0<t_1<t_2<\ldots<t_n$ and
$\|{Z}_{t_{j}}-{Z}_{t_{j-1}}\|_{\hir}>\varepsilon$ for $j=1,\ldots,n$. We say that ${Z}_\cdot(\omega)$ has
$\varepsilon$-oscillation infinitely often in $I$, if for every $n$, ${Z}_\cdot(\omega)$ has $\varepsilon$-oscillation
$n$ times in $I$. Let
\begin{align*}
&\Omega'=\{   \omega \in {\mathbb  M}^{[0,\infty)} \mid  \lim_{s\in {\mathbb  Q}, s \downarrow t}   {Z}_s(\omega)
\mbox{ and } \lim_{s\in {\mathbb  Q}, s \uparrow t}   {Z}_s(\omega) \mbox{ exist in } \hir \mbox{ for all }  t\geq 0 \},\\
&
A_{N,k}=\lkk  \omega \in {\mathbb  M}^{[0,\infty)} \mid  {Z}_t(\omega)  \mbox{ does not have } \frac{1}{k}-\mbox{oscillation infinitely often in }
 [0,  N] \cap {\mathbb  Q}    \rkk,\\
&\Omega''=\cap_{N=1}^{\infty}\cap_{k=1}^{\infty}A_{N,k}.
\end{align*}
Similarly as in \cite[Lemma ~11.2]{Sato}, it is seen that $\Omega''\subset \Omega'$. 
Define the event
$$
U(p,\varepsilon, I)=\{
 \omega\in \hirr \mid  {Z}_t(\omega) \mbox{ has } \varepsilon-\mbox{oscillation }
p  \mbox{ times in } I \}.
$$

\begin{lemma}\label{11.2}
For every  $\varepsilon>0$ we have 
$\underset{|t-s|\rightarrow  0} \lim {\M} \lk \| {Z}_t- { Z}_s \|_{\hir}
>\varepsilon \rk=0$ and $ {\M}(\Omega'')=~1$.
\end{lemma}
\proof
We assume that $|t-s|<1$ in this proof. 
Since $b_t$ and $\xi_t$ are translation invariant under $\M$, 
we have
$$
{\M}\lk \|  { Z}_t-  {Z}_s\|_{\hir}>\varepsilon  \rk 
=
{\M}\lk \|  { Z}_{t-s}-  {Z}_0\|_{\hir}>\varepsilon  \rk 
=(\one,  e^{-(t-s){L}_{\rm N}}\one_{B^c(\eps)})_{L^2(M)},
$$
where $B^c(\eps)=
\{ \|b_{t-s}-b_0\|_\BR^2+\|\xi_{t-s}-\xi_0\|_Q^2>\varepsilon^2\} $. Also, we set 
\begin{align*}
&K_2=\underset{y\in \BR} \sup\|\gr   (y,\cdot)\|_{L^2(\rm G)},\\
&C_V=\sup_{t\in[0,1]}\underset{y\in \BR}\sup\Ebb\left[
\exp\lk -2\int_0^t V(b_r) dr\rk \right],\\
&C_\vp=\sqrt 2 \exp(5\|\vp/\omega\|^2+4\|\vp/\sqrt\omega\|^2).
\end{align*}
We notice that $\grp \geq0$ and $\gr  \geq0$. 
Since 
\begin{align*}
\grp(y)&=
\Ebb^y[e^{-\int_0^t V(b_r) dr} \grp(b_t)]\leq
\Ebb^y[e^{-2\int_0^t V(b_r) dr}]^\han
\Ebb^y[ \grp(b_t)^2]^\han \\
&\leq C_V^\han 
\lk
\int_\BR |\grp(x+y)|^2 P_t(x) dx\rk^\han,
\end{align*}
where 
the distribution of $b_t$ on $\BR$ is given by $$P_t(x)=\frac{1}{(2\pi)^d} \frac{t}{\sqrt{|x|^2+t^2}}\int_\BR  e^{mt} e^{-\sqrt{(|x|^2+t^2)(|p|^2+m^2)}}dp.$$
We notice that 
$$P_t(x)=\lkk
\begin{array}{ll}<\infty &x=0\\
\leq |x|^{-(d+1)}(2\pi)^{-d}t e^{mt}\Gamma(d)&|x|>0
\end{array}
\right.$$
and $\|P_t\|_\infty<\infty$. 
Hence 
$\|\grp\|_\infty<\|P_t\|_\infty^\han  C_V^\han \|\grp\|_\LR<\infty$. 
We have
\begin{align*}
 {\M}\lk \|   {Z}_{t-s}-  {Z}_0\|_{\hir}>\varepsilon  \rk
&=
\int_{\hir }dy d{{\rm G}} \Ebb_{{\cal W}\times \G}^{(y,\xi)}
\left[      \grp (b_0)  \grp (b_{t-s}) e^{-\int_0^{t-s} V(\xx_r) dr}\right.\\
&\hspace{2cm}\left.
\gr  (b_0,\xi_0)  
\gr  (b_{t-s},\xi_{t-s})e^{-\int_0^{t-s} \tau_{\xx_r}\xi_r(\vpp)dr}
 \one_{B^c(\eps)}  \right]  e^{(t-s) E}.
 \end{align*}
We shall show that the right-hand side above converges to zero as $t-s\to 0$.
We reset $t-s$ by $t$, and we denote the integral with respect to $d{\rm G}d\G^\xi$  by 
\begin{align*}
X_t&=
\int_Q d{\rm G}
\Ebb_{\G}^{\xi}\left[
\gr  (b_0,\xi_0)
\gr  (b_t,\xi_t)
e^{-\int_0^{t} \tau_{\xx_r}\xi_r(\vpp)dr} 
 \one_{B^c(\eps)} \right]\\
&=
\Ebb_{\G}\left[
\gr  (b_0,\xi_0)
\gr  (b_t,\xi_t)
e^{-\int_0^{t} \tau_{\xx_r}\xi_r(\vpp)dr} 
 \one_{B^c(\eps)} \right]
  \end{align*}
for the notational simplicity. 
We also simply write 
$\Ebb^y$ for $\Ebb_{{\cal W}}^y$ and 
$\|\cdot\|_{L^p({\G})}=\|\cdot\|_p$ in this proof.  
Hence 
\begin{align*}
{\M}\lk \|   {Z}_{t}-  {Z}_0\|_{\hir}>\varepsilon  \rk
=
\int_{\BR}dy \grp (y)
\Ebb^y
\left[\grp (b_{t}) e^{-\int_0^{t} V(\xx_r) dr}X_t\right]e^{tE}.
 \end{align*}
Note that 
\begin{align*}
X_t&\leq 
\Ebb_\G[\gr  (b_0,\xi_0)\gr  (b_t,\xi_t)e^{-2\int_0^{t} \tau_{\xx_r}\xi_r(\vpp)dr} ]^\han 
\Ebb_\G[\gr  (b_0,\xi_0)\gr  (b_t,\xi_t)\one_{B^c(\eps)} ]^\han.
\end{align*}
Since 
\begin{align*}
\Ebb_\G[\gr  (b_0,\xi_0)\gr  (b_t,\xi_t)e^{-2\int_0^{t} \tau_{\xx_r}\xi_r(\vpp)dr} ]\leq 
\|\gr  (b_0)\|_2
\|\gr  (b_t)\|_2
e^{4t\gamma}
\leq 
2 K_2^2 
e^{4\gamma}
\end{align*}
Here the first  inequality follows from 
the fact $t<1$, 
the Baker-Campbel-Hausdorff formula
\cite[Proposition 1.29]{HL20}  and 
\cite[Corollary 1.88 (1.5.21)]{HL20}
with 
\begin{align}
\label{gamma}
\gamma=
\frac{5}{2}\|\vp/\omega\|^2+2\|\vp/\sqrt\omega\|^2.
\end{align}
Hence 
$$X_t\leq C_\vp K_2  A^\han,$$
where
$$A= 
\Ebb_\G[\gr  (b_0,\xi_0)\gr  (b_t,\xi_t)\one_{B^c(\eps)} ].
$$
We have 
\begin{align*}
&
\int_{\BR}dy \grp (y)
\Ebb^y[\grp (b_{t}) e^{-\int_0^{t} V(\xx_r) dr}X_t]\\
&\leq
C_\vp K_2 
\int_{\BR}dy \grp (y)
\Ebb^y
[\grp (b_{t}) e^{-\int_0^{t} V(\xx_r) dr} A^\han]\\
&\leq
C_\vp K_2 
\int_{\BR}dy \grp (y)
\Ebb^y[e^{-2\int_0^{t} V(\xx_r) dr}]^\han
\Ebb^y[\grp(b_t)^2 A]^\han\\
&\leq
C_V^\han C_\vp K_2 
\int_{\BR}dy \grp (y)
\Ebb^y[\grp(b_t)^2 A]^\han\\
&\leq
C_V^\han C_\vp K_2 
\int_{\BR}dy \grp (y)
\Ebb^y[\grp(b_t)^2 \|\gr(b_t,\xi_t)\one_{B^c(\eps)}\|_2\|\gr(b_0,\xi_0)\|_2]^\han\\
&\leq
C_V^\han C_\vp K_2^{3/2} 
\int_{\BR}dy \grp (y)
\Ebb^y[\grp(b_t)^2 \|\gr(b_t,\xi_t)\one_{B^c(\eps)}\|_2]^\han\\
&\leq
C_V^\han C_\vp K_2^{3/2} 
\|\grp\|
\lk
\int_{\BR}dy 
\Ebb^y[\grp(b_t)^2 \|\gr(b_t,\xi_t)\one_{B^c(\eps)}\|_2]\rk^\han.
\end{align*}
By the shift invariance (see \cite[Proposition 3.40]{LHB20}) and the reflection symmetry of $b_t$ and $\xi_t$, 
we have 
\begin{align*}
&C_V^\han C_\vp K_2^{3/2} 
\|\grp\|
\lk
\int_{\BR}dy 
\Ebb^y[\grp(b_t)^2 \|\gr(b_t,\xi_t)\one_{B^c(\eps)}\|_2]\rk^\han\\
&=C_V^\han C_\vp K_2^{3/2} 
\|\grp\|
\lk
\int_{\BR}dy 
\Ebb^y[\grp(b_0)^2 \|\gr(b_0,\xi_0)\one_{B^c(\eps)}\|_2]\rk^\han\\
&=C_V^\han C_\vp K_2^{3/2} 
\|\grp\|
\lk
\int_{\BR}dy 
\grp(y)^2 \Ebb^y[\|\gr(b_0,\xi_0)\one_{B^c(\eps)}\|_2]\rk^\han.
\end{align*}
By Schwartz inequality we have
\begin{align*}
&
\leq 
C_V^\han C_\vp K_2^{3/2} 
\|\grp\|
\|\grp\|^\han
\lk\int_{\BR}dy 
\grp(y)^2 \Ebb^y[\|\gr(b_0,\xi_0)\one_{B^c(\eps)}\|_2^2]\rk^{1/4}\\
&=
C_V^\han C_\vp K_2^{3/2} 
\|\grp\|^{3/2}
{\mathfrak M}(B^c(\eps))^{1/4}.
\end{align*}
Here 
\begin{align*}
{{\mathfrak M}}(
B^c(\eps))
=\int_{\mathbb M}\Ebb_{\W}^y \Ebb_\G^\xi[\one_{B^c(\eps)}] dm
\end{align*}
and 
$$dm=\grp(y)^2 \gr(y,\xi)^2 dy d{\rm G}.$$
Since $\|\grp\|_\infty<\infty$, 
$m$ is a finite measure.
Hence we have 
\begin{align}
\label{H5}
\M(\|Z_t-Z_0\|_{\mathbb M}>\eps)\leq
C_V^\han C_\vp K_2^{3/2}\|\grp\|^{3/2}e^{tE}
{\mathfrak M}(B^c(\eps))^{1/4}.
\end{align}
By the stochastic continuity of $b_t$
we can see that 
\begin{align*}
\Ebb^y
[B^c(\eps)]
=
\Ebb^y[\{\sqrt{
\|b_t-b_0\|_\BR^2+\|\xi_t-\xi_0\|_Q^2}>\eps\}]
\to 0
\end{align*}
as $t\to 0$.
The Lebesgue dominated convergence theorem yields that 
\begin{align}\label{H4}
\lim_{t\to0}
{{\mathfrak M}}(
B^c(\eps))
= 0.
\end{align}
By \kak{H5} it follows that 
$$
\lim_{t\to0}
\M(\|Z_t-Z_0\|_{\mathbb M}>\eps)=0.
$$
To see that $\M (\Omega'')$=1, it suffices to show that $\M (A_{N,k}^c)=0$ for any fixed $N$ and $k$. We have
\begin{align*}
&\M (A_{N,k}^c)
=\M \lk\lkk  {Z}_t  \mbox{ has }\frac{1}{k}-\mbox{oscillation infinitely often in } \big [0,  N\big] \cap {{\mathbb  Q}}    \rkk \rk \\
& \leq  \sum_{j=1}^{l}  \M  \lk \lkk  {Z}_t\mbox{ has
 }\frac{1}{k}-\mbox{oscillation infinitely often in }
\bigg [\frac{j-1}{l} N, \frac{j}{l}N\bigg] \cap {\mathbb  Q}    \rkk \rk \\
&=\sum_{j=1}^{l}\underset{p\rightarrow  \infty}\lim \M  \lk U
\lk p, \frac{1}{k},\bigg [\frac{j-1}{l} N, \frac{j}{l}N\bigg]    \cap {\mathbb  Q}\rk
  \rk .
\end{align*}
We enumerate as $\{t_1,\ldots,t_n,\ldots\} = [\frac{j-1}{l} N, \frac{j}{l}N]    \cap {\mathbb  Q} $.
Fix $p$. 
Thus
$$
\M  \lk U
\lk p, \frac{1}{k},\bigg [\frac{j-1}{l} N, \frac{j}{l}N\bigg]    \cap {\mathbb  Q}\rk
  \rk
=\lim_{n\to\infty}\M   \lk 
U \lk p, \frac{1}{k},  \{t_1,\cdots,t_n\}\rk  \rk .$$
Since 
$$U_n=U \lk p, \frac{1}{k},  \{t_1,\cdots,t_n\}\rk=
\bigcup_{\{j_1,\ldots,j_p\}\subset \{1,\ldots,n\}\atop
j_1<\ldots<j_p}
\{\omega\in{\mathbb M}\ ; \ 
\|Z_{t_{j_i}}-Z_{t_{j_i-1}}\|>\frac{1}{k}, i=1,\ldots,p\}$$
Then by  the translation invariance  
of $b_t$ and $\xi_t$ under $\M$ we can assume that $0\leq t_1,\ldots,t_n
\leq N/l$, 
we obtain
\begin{align*}
&\M   (U_n)\\
&= e^{N E/l}
\int_{\hir} dyd{\rm G}\Ebb^{y,\xi}_{\W\times\G}
\left[\grp(b_0)\grp(\xx_{N/l})
e^{-\int_0^{N/l} V(\xx_r) dr}
\gr(b_0,\xi_0)     
\gr(\xx_{N/l},\xi_{N/l})    
e^{-\int_0^{N/l} \tau_{\xx_r}\xi_r(\vpp)dr}
  \one_{U_n}\right].
  \end{align*}
Hence in the same estimate preceding \kak{H4} we have
\begin{align}
\label{darui2}
\M (U_n)
\leq 
C_V^\han C_\vp K_2^{3/2} 
\|\grp\|^{3/2}
{\mathfrak M}(U_n)^{1/4}
e^{(N/l)  E}.
\end{align}
By \cite[Lemma 11.4]{Sato}, furthermore we have
\eq{kore}
{\mathfrak M}(U_n)
\leq
\lk
\sup_{s,t\in[0,N]\atop t-s\in[0,N/l]}
{\mathfrak M}
\left (\sqrt{\|b_t-b_s\|_\BR^2+\|\xi_t-\xi_s\|^2_Q}>\frac{1}{4k} \right)\rk ^p.
\en
Moreover, by stochastic continuity of $((\xx_t,\xi_t))_{t\geq  0}$ under $\M$,
we can furthermore prove uniform stochastic continuity, i.e.,
\eq{us}
\sup_{s,t\in[0,N]\atop
t-s\in[0,N/l]}{\mathfrak M}\left (\sqrt{\|b_t-b_s\|_\BR^2+\|\xi_t-\xi_s\|^2_Q}\geq\frac{1}{4k} \right)\to 0
\en
as
$l\to\infty$  in Lemma \ref{satobook} below.
Then for arbitrary $\eps>0$, 
${\mathfrak M}(U_n)\leq \eps^p$ for some large $l$. 
Hence 
\begin{align*}
\M (A_{N,k}^c)
\leq
C_V^\han C_\vp K_2^{3/4} 
\|\grp\|^{3/2}
\sum_{j=1}^{l}\underset{p\rightarrow  \infty}\lim 
 \eps^{p/4}
=0.
\end{align*}
Then the proof is complete. 
\qed

\begin{lemma}\label{satobook}
 \kak{us} holds.
\end{lemma}
\proof
For notational simplicity we write $X'_s=(b_s,\xi_s)$.
Fix $a>0$. For any $t$ there exists $\delta_t>0$ such that
 ${\mathfrak M}(|X'_t-X'_s|\geq \varepsilon/2) \leq a/2$
 for $|t-s|<\delta_t$ by stochastic continuity. Let $I_t=(t-\delta_t/2,t+\delta_t/2)$.
 Since $I_t$ is compact, there exists a finite covering $I_{t_j}$, $j=1,\ldots,n$, such that
 $\cup_{j=1}^n I_{t_j}\supset[0,N]$. Let $\delta=\min_{j=1,...,n}\delta_{t_j}$.
 If $|s-t|<\delta$ and $s,t\in [0,N]$, then $t\in I_{t_j}$ for some $j$, hence $|s-t_j|<\delta_{t_j}$ and
 $$
 {\mathfrak M} (|X'_t-X'_s|\geq \varepsilon)
 \leq {\mathfrak M} (|X'_t-X'_{t_j}|\geq \varepsilon)
 +{\mathfrak M}(|X'_{t_j}-X'_s|\geq \varepsilon) <a.
 $$
Hence the lemma follows.
\qed

\begin{lemma}\label{cadlag}
The process $({Z}_t)_{t\geq 0}$ has a c\`adl\`ag version with respect to ${\M}$.
\end{lemma}
\proof
Let $(Z_t')_{t\geqslant0}$ be a c\`adl\`ag process defined by
\begin{align}
\label{mod}
{Z}'_t(\omega)= \left\{\begin{array}{ll}
\underset {s\in {\mathbb  Q}, s \downarrow t}\lim  { Z}_s(\omega) \quad& \omega\in\Omega'',\\
0 \quad& \omega\notin\Omega''.
\end{array}
\right.
\end{align}
By Lemma \ref{11.2}   the process $({Z}_t)_{t\geq 0}$  is stochastically continuous, which implies that
there exists a sequence $\seq s$ such  that
$
\underset {s_n\in {\mathbb  Q}, s_n \downarrow t}\lim   {Z}_{s_n}(\omega) ={Z}_t(\omega)
$
for $\omega \in \Omega '''=\hir^{[0,\infty)}\setminus N_t$ with a null set $N_t$. We also see by
the definition of  $\proo{Z'}$  that
$
\underset {s_n\in{\mathbb  Q}, s_n \downarrow t}\lim   {Z}_{s_n}(\omega) ={Z}'_t(\omega)
$
for $\omega \in \Omega ''$, and ${\M}(\Omega'')=1$ by Lemma  \ref{11.2}. 
Hence for
every $t$ it follows that ${Z}_t(\omega) ={Z}'_t(\omega) $ for $\omega \in \Omega''\cap \Omega '''$,
and $ {\M}(\Omega''\cap \Omega ''')=1$. 
Thus $\proo {Z'}$ is a c\`adl\`ag  version of
$\proo Z$.
\qed

We denote the c\`adl\`ag version of $( Z_t)_{t\geq 0}$ by $(\bar{ {Z_t}})_{t\geq 0}$, and the set of
$\hir$-valued c\`adl\`ag paths by 
$$\DD_\hir ^+ =D(\RR^+, \hir).$$
Note that $(\bar{ {Z_t}})_{t\geq 0}$
is a random  process on the probability space $(\hir ^{[0,\infty)}, \s({\mathcal  A}), {\M} )$, 
and the map
$$
\bar{{ Z_\cdot}}:(\hir ^{[0,\infty)}, \s({\mathcal  A}), {\M} )\to (\DD_\hir ^+ , \tilde \C^+   )
$$
is measurable, where $\tilde \C^+  $ denotes the $\s$-field generated by cylinder sets. This map induces the image
measure $\P _+={\M} \circ {\bar{ {Z_\cdot}}} \f$ on $(\DD_\hir ^+ , \tilde \C^+  )$. 
The coordinate process
$({ X_t^+})_{t\geq0}$ on $(\DD_\hir ^+ , \tilde \C^+    ,\P _+)$ satisfies $\bar{ {Z_t}}\stackrel{\rm d}{=} X^+_t$ for
$t\geq0$.  Let $(y,\xi)\in\hir $ and define a regular conditional probability measure on $(\DD_\hir ^+,\tilde \C^+  )$ by
$$
\P _+^{(y,\xi)}(\cdot)=\P _+(\cdot \mid   X^+_{0}=(y,\xi)).
$$
Since the distribution of $X^+_{0}$ is ${\rm M}$, we see that
$
\P_+  (A)=\int_{\hir } \Ebb_{\P_+} ^{(y,\xi)}[\one _{A}] d{\rm M}(y,\xi)$. 
Using $\P_+^{(y,\xi)}$ we have
\begin{align}
\label
{pos}
&(f, e^{-t\LN }g)_{\LM}=
\int _\hir \Ebb_{\P_+}^{(y,\xi)}[\bar f(X^+_0) g(X^+_t)]
d{\rm M}(y,\xi),\\
\label{semigroup3}
&( e^{-t\LN }g) (y,\xi)=
 \Ebb_{\P_+}^{(y,\xi)} [g(X^+_t)].
\end{align}

\begin{lemma}
\label{markov}
The random process $(X^+_{t})_{t\geqslant0 }$ on $(\DD_\hir ^+,\tilde \C^+    ,\P_+^{(y,\xi)})$ has the Markov
property with respect to the natural filtration $\sigma(X^+_{s}, 0\leq s\leq t)$.
\end{lemma}
\proof
In this proof we set $z=(y,\xi), z_j=(y_j,\xi_j)\in \hir $ for notational simplicity.
Let
\begin{align}
p_{t}(z, A)=(e^{-t\LN }\one_A)(z), \quad A\in
\B_\hir .
\end{align}
Notice that
$
p_{t}(z, A)=\Ebb_{\P_+}^{z}[\one _{A}(X^+_{t})]=\Ebb_{\P_+ }[\one _{A}(X^+_{t})| X^+_{0}=z]$. 
We show that $p_{t}(z, A)$  is a probability transition kernel, i.e., 
(a) $p_{t}(z, \cdot)$ is a probability measure on $\B (Q)$, 
(b) the function $z\mapsto p_{t}(z, A)$ is  Borel measurable, and
(c) the Chapman-Kolmogorov identity $\int_\hir p_t(z,A) p_s(x,dz)=p_{s+t}(x,A)$ is satisfied. 
First
note that by \kak{pos} it is easy to see that  $e^{-t\LN }$ is positivity improving. For every
function $f\in\hhh $ such that $0\leq f\leq\one$, we have
$$
( e^{-t\LN }f)(z)=\Ebb_{\P_+}^{z}[f(X^+_{t})]\leq
\Ebb_{\P_+}^{z}[\one]=1.
$$
Thus $0\leq e^{-t\LN }f\leq\one$  and $e^{-t\LN }\one=\one$, and (a)-(b) follow. We can also show that the
finite dimensional distributions are given by
\begin{align*}
\Ebb_{\P_+}\left[\prod_{j=1}^{n}\one_{A_{j}}(X^+_{t_{j}})\right] =\int_{\hir ^{n}}
\prod_{j=1}^{n}\one_{A_{j}}(z_j)\prod_{j=1}^{n}p_{t_{j}-t_{j-1}}(z_{j-1}, dz_j).
\end{align*}
This implies (c), and thus $(X^+_{t})_{t\geq0}$ is a Markov process by 
\cite[Proposition 2.17]{LHB}.
\qed

Now we extend $(X^+_t)_{t\geq 0}$ to a Markov process indexed by the real line $\RR$.  
Consider the product
probability space 
$(\hat\XX_\hir,\hat \C ,\hat  \P^{(y,\xi)})$ with $\hat\XX_\hir=\DD_\hir^+\times\DD_\hir^+ $,
$\hat \C     =\tilde \C^+    \otimes \tilde \C^+   $ and $\hat  \P^{(y,\xi)}=\P_+^{(y,\xi)}\otimes\P_+^{(y,\xi)}$, and let
$\pro {\hat  X}$ be the random  process on the product space defined by 
$$
\hat   {X}_t(\omega)= \left\{\begin{array}{rl}
 X_t^+(\omega_1) \quad t\geq 0,\\
X_{-t}^+(\omega_2) \quad t\leq0,
\end{array}
\right.\quad \omega=(\omega_1,\omega_2)\in\hat\XX_\hir.$$ 
\begin{lemma}
\label{hiroc}
It follows that
\begin{description}
\item[\rm (i)]
${\hat { X}_0}={(y,\xi)}$   $\hat  \P^{(y,\xi)}$-almost surely,
\item[\rm (ii)]
${\hat { X}_t},t\geq 0$ and ${\hat { X}_s}, s<0$ are independent,
\item[\rm (iii)]
${\hat { X}_t}\stackrel{\rm d}{=}{\hat { X}_{-t}}$ for all $t\in\RR$,
\item[\rm (iv)]
$({\hat { X}}_t)_{t\geq 0}$ and $({\hat { X}}_t)_{t\leq0}$ are Markov processes for the filtrations
$\s(\hat { X}_s,0\leq s\leq t)$ and $\s( \hat { X}_s,t\leq s\leq 0)$, respectively,
\item[\rm (v)]
for $f_0,\ldots,f_n\in {\LM}$ and $-t=t_{0}\leq t_{1}\leq\ldots\leq t_{n}=t$, we have
\eq{p33}
\Ebb_{\hat  {\P }_+}\left[ \prod_{j=0}^{n}f_{j}(\hat { X}_{t_{j}})\right] =\lk \bar  f_{0},
e^{-(t_{1}+t){L}_{\rm N} }f_{1}\cdots f_{n-1}e^{-(t-t_{n-1}){L}_{\rm N}  }f_{n}\rk _{\LM}.
\en
\end{description}
\end{lemma}
\proof
(i)-(iii) are straightforward. 
(iv) follows from Lemma \ref{markov}. 
(v) follows from \kak{p2} and a simple
limiting argument.
\qed

\medskip
\noindent
{\it Proof of Theorem \ref{pp11}:} We show that the random process $\pro {\hat  { X}}$ defined on
$(\hat\XX_\hir, \hat  \C^+, \hat {\P}^{(y,\xi)})$ is a $P(\phi)_1$-process associated with
$((\hir, \B_\hir,  {\rm M}), {L}_{\rm N})$. The Markov property, reflection symmetry and shift invariance
follow from Lemma \ref{hiroc}. The c\`adl\`ag property of  $t\mapsto \hat  X_t$, $t\geq0$, (resp.
c\`agl\`ad property of $t\mapsto \hat  X_t$, $t\leq 0$) was shown in Lemma~\ref{cadlag}. 
Thus the map
$\hat { X}_\cdot: (\hat\XX_\hir, \hat  \C^+ , \hat {\P}^{(y,\xi)})\to (\DD_\hir ,\tilde \C_\hir)$ is
measurable and the image measure $\tilde \P^{(y,\xi)}=\hat  {\P}^{(y,\xi)}\circ \hat {X}_\cdot\f$ defines a
probability measure on $(\DD_\hir ,\tilde \C_\hir)$. 
Hence the coordinate process $\pro {{\tilde X}}$ on 
$(\DD_\hir ,\tilde \C_\hir, \tilde\P^{(y,\xi)})$
satisfies $ \tilde X_t\stackrel{\rm d}{=} \hat  { X}_t$, and is a $P(\phi)_1$-process associated with 
$\left((\hir, \B_\hir , {\rm M}), {L}_{\rm N}\right)$.
\qed

\subsection{$P(\phi)_{1}$-process associated with the non-relativistic  Nelson model}
We consider the non-relativistic case. As before, formulae \eqref{rmm}-\eqref{nelgst} hold similarly
for the non-relativistic  setting. Let now 
$$\XXQ=C(\RR,\hir)$$ be the set of continuous paths with values in $\hir$ and
indexed by the real line $\RR$, and $\C_\hir$ the $\s$-field generated by the cylinder sets. A second main result
of this section is the following.
\bt{pp1}\TTT{$P(\phi)_1$-process associated with the non-relativistic  Nelson Hamiltonian}
\label{classnel}
Let $(y,\xi)\in \hir$. Then the following hold.
\begin{description}
\item[\rm (i)]
There exists  a probability measure $\P^{(y,\xi)}$ on the measurable space $(\XXQ,  \C_\hir)$ such that the coordinate
process $\pro X$ on $(\XXQ,  \C_\hir,\P^{(y,\xi)})$ is a $P(\phi)_1$-process associated with 
$((\hir, \B_\hir , {\rm M}), \LN)$.
\item[\rm (ii)]
The function $t\mapsto X_t$ is almost surely continuous.
\end{description}
\et

The construction involving the counterpart of the probability measure \eqref{mmeasure} is similar as in 
Theorem~\ref{relnel}, and we leave the details to the reader. We only discuss path continuity, which makes a substantial
difference from the previous case. We write $Z_t=(x_t,\xi_t)$, where $x_t\in\BR$ and $\xi_t\in Q$ are the
coordinate processes $x_t(\omega)=\omega_1(t)$ and $\xi_t(\omega)=\omega_2(t)$ for all $t\geqslant0$ and $\omega=
(\omega_1,\omega_2)\in \hir^{[0,\infty)}$.

\begin{lemma}
\label{cont}
The random process $(Z_{t})_{t\geq 0 }$ on $((\hir)^{[0,\infty)}, \s(\ms A))$ has a continuous version.
\end{lemma}
\proof
 Define
$\|Z_t\|_{\hir}=\sqrt{\|x_t\|^2_{\BR}+\|\xi_t\|_Q^2}$. 
Let
\begin{align*}
&C_2=\underset{x\in \BR} \sup\|\gr   (x,\cdot)\|_{L^2(\rm G)},\\
&D_V=\sup_{t\in[0,1]}\underset{y\in \BR}\sup\Ebb\left[
\exp\lk -4\int_0^t V(B_r) dr\rk \right],\\
&C_V=\sup_{t\in[0,1]}\underset{y\in \BR}\sup\Ebb\left[
\exp\lk -2\int_0^t V(B_r) dr\rk \right],\\
&D_\vp= 2\exp(\frac{5}{2}\|\vp/\omega\|^2+2\|\vp/\sqrt\omega\|^2),\\
&C_\vp= \sqrt 2\exp(5\|\vp/\omega\|^2+4\|\vp/\sqrt\omega\|^2).
\end{align*}
We notice that $\grp \geq0$ and $\gr  \geq0$. 
Using the Kolmogorov-\v{C}entsov theorem, 
the estimate
\eq{kc}
\Ebb_\M [\|  Z_{t}-Z_{s}\| _{\hir }^{\alpha}]\leq D|t-s|^{\beta+1},\quad s,t\in [S,T]
\en
with a suitable $D>0$, $\alpha\geq1$ and $\beta>0$  implies that 
$(Z_{t})_{s\leq t\geq T}$ has a continuous version. 
Since
$$
\|Z_t\|_{\hir}^4\leq 2(\|x_t\|^4_{\BR}+\|\xi_t\|_Q^4),
$$
and the translation invariance of $x_t$ and $\xi_t$ under $\M$ 
it suffices to prove the bounds
\begin{align}
&\label{33}
\Ebb_\M [\|  x_{t}-x_{s}\| _{\BR}^{4}]\leq D_1|t-s|^{2},\\
&\label{34}
\Ebb_\M [\|  \xi_{t}-\xi_{s}\| _{Q}^{4}]\leq D_2|t-s|^{4}
\end{align}
for $s,t\in [0,1]$. 
To obtain \kak{33}, recall the moments formula $\Ebb[|B_t-B_s|^{2n}]= K_n |t-s|^n$, with a constant $K_n$ for $n\geq0$. Let
$x_t=(x_t^1,\ldots,x_t^d)$. 
We have for all $1\leq i,j\leq d$ that
\begin{align*}
&\Ebb_\M [(x_t^i)^n  (x_s^j)^m]\\
&=\int_\BR \!\! dx\Ebb_{\W\times \G}^x\left[
(B_{0}^i)^n (B_{T}^j)^m  \grp(B_0)\grp(B_{T})
e^{-\int_0^{T} V(B_r) dr}
\gr(B_0,\xi_0)
\gr(B_{T},\xi_{T})
e^{-\int_0^{T} \tau_{B_r}\xi_r(\vpp ))dr}
\right]e^{TE}
\end{align*}
with $T=t-s$. 
We set $X=|B_{0}- B_{t-s}|$ and write $\Ebb^x$ for 
$\Ebb_\W^x$. 
Thus we have
\begin{align*}
&\Ebb_\M [|x_t- x_s|^4]\\
&=
\int_\BR \!\! dx\Ebb_{\W\times \G}^x\left[ 
X^4\grp(B_0)
\grp(B_{t-s})
e^{-\int_0^{t-s} V(B_r) dr}
\gr(B_0,\xi_0)
\gr(B_{t-s},\xi_{t-s}) 
e^{-\int_0^{t-s} \tau_{B_r}\xi_r(\vpp )} 
 \right]
e^{|t-s|E}\\
&
=
\int_\BR \!\! dx\Ebb^x\left[X^4 \grp(B_0)
\grp(B_{t-s}) e^{-\int_0^{t-s} V(B_r) dr} Y_t
\right]
e^{|t-s|E},
\end{align*}
where 
$$Y_t=
\Ebb_\G[
\gr(B_0,\xi_0)
\gr(B_{t-s},\xi_{t-s}) 
e^{-\int_0^{t-s} \tau_{B_r}\xi_r(\vpp )}].
$$
By 
the fact $|t-s|<1$, 
the Baker-Campbel-Hausdorff formula
\cite[Proposition 1.29]{HL20}  and 
\cite[Corollary 1.88 (1.5.21)]{HL20}, 
we have 
$$Y_t\leq C_2^2D_\vp.$$
Hence 
\begin{align*}
&\leq
C_2^2 D_\vp  e^{|t-s|E}
\int_\BR \!\! dx \Ebb^x\left[
\grp(x) \grp(B_{t-s})
X^4 
 e^{-\int_0^{t-s} V(B_r) dr} \right]\\
&\leq
C_2^2 D_\vp  e^{|t-s|E}
\lk
\int_\BR dx \Ebb^x[|\grp(B_{t-s})|^2]\rk^\han 
\lk 
\int_\BR dx \Ebb^x[
|\grp(x)|^2 X^{8}  e^{-2\int_0^{t-s} V(B_r) dr} ]\rk^\han\\
&\leq 
C_2^2 D_\vp e^{|t-s|E}
\|\grp\|
\lk 
\int_\BR dx |\grp(x)|^2 
\Ebb^x[
X^{16}]^\han
\Ebb^x[
 e^{-4\int_0^{t-s} V(B_r) dr} ]^\han\rk^\han\\
&\leq 
C_2^2 D_\vp  e^{|t-s|E}
\|\grp\|^2
D_V^{1/4}
\Ebb^0[
X^{16}]^{1/4}
\end{align*}
and 
we have 
\begin{align*}
\Ebb_\M [|x_t- x_s|^4]
&\leq
C_2^2 
\|\grp\|^2
D_\vp  e^{|t-s|E}
D_V^{1/4}
K_{16}^{1/4}|t-s|^2.
\end{align*}
Thus \kak{33} follows.

Next we prove \kak{34}.
Let $f\in \ms M_{+2}$.
In the same way as in the proof of \kak{33} we have
\begin{align*}
&\Ebb_\M [\xi_t(f)^n  \xi_s(f)^m]\\
&=\int_\BR \!\! dx\Ebb_{\W\times \G}^x\left[
\xi_0(f)^n \xi_T(f)^m  \grp(B_0)
\grp(B_{T})
e^{-\int_0^{T} V(B_r) dr}\gr(B_0,\xi_0)
\gr(B_{T},\xi_{T})
e^{-\int_0^{T} \tau_{B_r}\xi_r(\vpp ))dr}
\right]e^{|T|E}
\end{align*}
with $T=t-s$.
Let $X=|\xi_0(f)-\xi_{t-s}(f)|$.
Hence
\begin{align*}
&\Ebb_\M [|\xi_t(f)-\xi_s(f)|^{2n}]\\
&=
\int_\BR \!\! dx\Ebb_{\W\times \G}^x
\left[
\grp(B_0)
\grp(B_{t-s})X ^{2n} e^{-\int_0^{t-s} V(B_r) dr} \gr(B_0,\xi_0)\gr(B_{t-s},\xi_{t-s})
e^{-\int_0^{t-s} \tau_{B_r}\xi_r(\vpp ))dr}
\right]e^{|t-s|E}\\
&=
\int_\BR \!\! dx\Ebb^x
\left[
\grp(B_0)
\grp(B_{t-s})e^{-\int_0^{t-s} V(B_r) dr}
\Ebb_\G\left[X ^{2n}  \gr(B_0,\xi_0)\gr(B_{t-s},\xi_{t-s})
e^{-\int_0^{t-s} \tau_{B_r}\xi_r(\vpp ))dr}
\right]
\right]e^{|t-s|E}.
\end{align*}
We can estimate the term $\Ebb_\G[\ldots]$ above  
by Schwartz inequality.  
We have 
\begin{align*}
&\Ebb_\G\left[X ^{2n}  \gr(B_0,\xi_0)\gr(B_{t-s},\xi_{t-s})
e^{-\int_0^{t-s} \tau_{B_r}\xi_r(\vpp ))dr}
\right]\\
&\leq 
\Ebb_\G[X ^{4n}\gr(B_0,\xi_0)\gr(B_{t-s},\xi_{t-s})]^\han
\Ebb_\G[e^{-2\int_0^{t-s} \tau_{B_r}\xi_r(\vpp ))dr} \gr(B_0,\xi_0)\gr(B_{t-s},\xi_{t-s})]^\han\\
&\leq C_2C_\vp
\Ebb_\G[X ^{4n} \gr(B_0,\xi_0)\gr(B_{t-s},\xi_{t-s})]^\han\\
&\leq C_2C_\vp 
\Ebb_\G[X ^{8n} |\gr(B_0,\xi_0)|^2]^{1/4}
\Ebb_\G[ |\gr(B_{t-s},\xi_{t-s})|^2]^{1/4}\\
&\leq C_2^{3/2}C_\vp
\Ebb_\G[X ^{8n} |\gr(B_0,\xi_0)|^2]^{1/4}
\end{align*}
Then 
\begin{align*}
&\Ebb_\M [|\xi_t(f)-\xi_s(f)|^{2n}]\\
&\leq 
C_2^{3/2}
C_\vp
e^{|t-s|E}
\int_\BR \!\! dx\Ebb^x
\left[
\grp(B_0)
\grp(B_{t-s})e^{-\int_0^{t-s} V(B_r) dr}
\Ebb_\G[X ^{8n} |\gr(B_0,\xi_0)|^2]^{1/4}
\right]\\
&
\leq
C_2^{3/2}
C_\vp
e^{|t-s|E}
\lk
\int_\BR \!\! dx
|\grp(x)|^2
\Ebb_\G[X ^{8n} |\gr(B_0,\xi_0)|^2]^\han
\rk^\han
\\
&\hspace{3cm}\times \lk
\int_\BR \!\! dx
\Ebb^x
[
|\grp(B_{t-s})|e^{-\int_0^{t-s} V(B_r) dr}
]^2\rk^\han
\\
&\leq
C_2^{3/2}
C_\vp
e^{|t-s|E}
\lk
\int_\BR \!\! dx
|\grp(x)|^2
\rk^{1/4}
\lk
\int_\BR \!\! dx
|\grp(x)|^2
\Ebb_\G[X ^{8n} |\gr(B_0,\xi_0)|^2]
\rk^{1/4}
\\
&\hspace{3cm}\times\lk
\int_\BR \!\! dx
\Ebb^x
[|\grp(B_{t-s})|^2]
\Ebb^x
[e^{-2\int_0^{t-s} V(B_r) dr}
]\rk^\han
\\
&\leq
C_2^{3/2}
\|\grp\|^{3/2}
C_\vp
C_V^\han
e^{|t-s|E}
\mathfrak{M}(X ^{8n})^{1/4}.
\end{align*}
Here 
\begin{align*}
{{\mathfrak M}}(
X ^{8n})
=\int_{\mathbb M}\Ebb_{\W}^y \Ebb_\G^\xi[X ^{8n}] dm
\end{align*}
with 
$dm=\grp(x)^2 \gr(x,\xi)^2 dy d{\rm G}$. 
It is easy to estimate 
$\mathfrak{M}(X ^{8n})$. 
Let $N$ be the number operator. 
Then 
$\|e^{\beta N} \gr\|<\infty$ for any $\beta>0$ and 
$$ \||\xi_t-\xi_0|^{2n} \gr\|\leq 
C_n \|(N+\one)^n\gr\| \|(1-e^{-t\omega})f\|^{2n}$$ with some constant $C_n$. 
Hence 
\begin{align*}
\mathfrak{M}(X ^{8n})
&\leq 
\|\grp\|_\infty^2 \||\xi_{t-s}(f)-\xi_0(f)|^{4n}\gr\|^2\\
&\leq 
C_{2n}^2\|\grp\|_\infty^2 
\|(N+\one)^{2n}\gr\|^2
\|(1-e^{-|t-s|\omega})f\|^{8n}\\
&\leq 
C_{2n}^2 \|\grp\|_\infty^2 
\|(N+\one)^{2n}\gr\|^2
|t-s|^{8n} \|f\|_{\ms M_{+2}}^{8n}. 
\end{align*}
Finally we have the bound 
\begin{align*}
&\Ebb_\M [|\xi_t(f)-\xi_s(f)|^{2n}]\\
&\leq
C_2^{3/2}
\|\grp\|^{3/2}
C_\vp
C_V^\han 
e^{|t-s|E}
\mathfrak{M}(X^{8n})^{1/4}\\
&\leq
C_2^{3/2}
\|(N+\one)^{2n}\gr\|^\han 
\|\grp\|^{3/2}
\|\grp\|_\infty^\han 
C_\vp
C_V^\han
C_{2n}^\han  
e^{|t-s|E}
|t-s|^{2n} \|f\|_{\ms M_{+2}}^{2n}
\end{align*}
Hence
\eq{tg}
\Ebb_\M \left[\frac{|\xi_t(f)-\xi_s(f)|^{2n}}{\|f\|^{2n}_{\ms M_{+2}}}\right]
\leq D
|t-s|^{2n}
\en
with a constant 
$$D=C_2^{3/2}
\|(N+\one)^{2n}\gr\|^\han 
\|\grp\|^{3/2}
\|\grp\|_\infty^\han 
C_\vp
C_V^\han
C_{2n}^\han
e^{|t-s|E}.$$ 
Since $\|\xi_t-\xi_s\|_Q=\sup_{f\not=0}|\xi_t(f)-\xi_s(f)|/\|f\|_{\ms M_{+2}}$,
for every $\eps>0$ there exists $f_\eps\in  \ms M_{+2}$ such that
$\|\xi_t-\xi_s\|_Q\leq |\xi_t(f_\eps)-\xi_s(f_\eps)|/\|f_\eps\|_{\ms M_{+2}}+\eps$. Thus
together with  \kak{tg} we have
\eq{tg2}
\Ebb_\M \left[{|\xi_t-\xi_s|^4}\right]-\eps
\leq D |t-s|^4,
\en
and thus \kak{34} follows.
\qed

\section{Functional central limit theorems}
\subsection{FCLT for additive functionals}
In this section we obtain a FCLT
 for  the Nelson Hamiltonians as an immediate consequence.
Let $\proo B$ be one-dimensional  Brownian motion. We first define somes notions that we need in this section.
\begin{definition}
\begin{enumerate}
\item[(i)]  
A process is said to be stationary if $( Z_{t_{1}}, Z_{t_{2}},\ldots, Z_{t_{n}})$ has the same
distribution as $( Z_{\tau+t_{1}}, Z_{\tau+t_{2}},\ldots, Z_{\tau+t_{n}})$ for all $t_{1},\ldots,t_{n},\tau\in\mathbb{R}$.
\item [(ii)] 
Let $(E,\rho)$ be a $\sigma$-finite measure space. A semigroup, $P_t$ of bounded operators is called ergodic if for all $ f\geq0,g\geq0$ but $f\not\equiv 0,g\not\equiv 0$, $f,g\in L^2(E)$ there is a $t\geqslant0$ with $\lk f, P_t g\rk\neq 0.$
\end{enumerate}
\end{definition}
\begin{remark}\label{ergo}
Let $e^{-tL}$ be a semigroup  of self-adjoint positivity preserving operators. Then the semigroup is ergodic if and only if it is positivity improving.
\end{remark}

A fundamental tool used in this section is as follows: 
\bp{kv}\TTT {\cite[Theorem  2.1]{bha82}}
Let $(\Omega,\F, \proo \F, P)$ be a filtered probability space and  
$(\mathbb{R},\F_{\mathbb{R}},\pi)$ a probability space. 
Let $Z=\proo Z$ be  an $\mathbb{R}$-valued Markov process  with respect to $\proo \F$.
Assume that 
\bi
 \item[\rm (i)]  $f$ is a non-constant bounded function for $f\in D(L^{-1/2})$, 
 and $\int_\RR  f d\pi =0$, 
\item[\rm (ii)] 
 $Z$ is stationary  and ergodic Markov process with a positive self-adjoint generator $L$ acting in $L^2(\mathbb{R},d\pi)$,
\item[\rm (iii)] $\pi$ is an invariant initial distribution  for $Z$; 
$\Ebb_P[f(Z_t)]=\int_{\mathbb{R}} \Ebb_{P^z}[f(Z_t)]d\pi(z)=\int_{\mathbb{R}} f(z) d\pi(z)$=$\Ebb_P[f(Z_0)]$.
\ei
 Then  the random process 
 $\frac{1}{\sqrt{s}} \int_0^{st} f(Z_r)dr $ weakly converges to $\SS B_t
 $ as $s \to \infty$, where  $\SS=2\|L^{-1/2} f\|^2$.
\ep

We define probability measures on sets 
$\hir\times \DD_\hir$ and $\hir\times\XX_\hir$ by 
\eq{P}
d \tilde \P=d\tilde \P^{(y,\xi)}d{\rm M}(y,\xi),\quad 
d \P=d \P^{(y,\xi)}d{\rm M}(y,\xi),
\en
 respectively. 
We constructed two $P(\phi)_1$-processes, 
$\pro{{\tilde X}}$ on $(\DD_\hir,\tilde\C_\hir,\tilde \P^{(y,\xi)})$ associated with relativistic case 
and 
$\pro{{X}}$ on $(\XX_\hir,\C_\hir,\P^{(y,\xi)})$ for non-relativistic case. 
Our process $\pro{{\tilde X}}$ (resp. $\pro{{X}}$) is Markov with semigroup $e^{-t\LN}$ and 
\begin{align*}
&p_{t}(z, A)=(e^{-t\LN }\one_A)(z)=\Ebb_{\tilde \P}^{z}[\one _{A}(\tilde X_{t})]
=\Ebb_{\tilde \P }[\one _{A}(\tilde X_{t})| \tilde X_{0}=z],\\
&{\rm (resp}. \ \ \ p_{t}(z, A)=(e^{-t\LN }\one_A)(z)=\Ebb_{\P}^{z}[\one _{A}(X_{t})]=\Ebb_{\P }[\one _{A}(X_{t})| X_{0}=z])
\end{align*}
for 
$z=(y,\xi)$ and $A\in
\B(\BR)\times\B (Q)$. 
We define the additive functional
\begin{align}\label{clt}
\FF{t} =\int_{0}^{t}f(\tilde X_{s})ds \quad 
({\rm resp.} \ \ \FFF{t} =\int_{0}^{t}f(X_{s})ds).
\end{align}
The purpose of this section is to show 
the limit 
$\FF{st}/\sqrt s\to \SS(f)B_t$ (res. $\FFF{st}/\sqrt s\to \SS(f)B_t$) as $s\to\infty$ in some sense.

\begin{lemma}\label{inv}
The probability measure ${\rm M}$  associated to the $P(\phi)_1$-processes 
$\proo {\tilde X}$ (resp. $\proo X$)  is an invariant measure  
 in the sense that
$$\int_\hir  \Ebb_{\P}^{(y,\xi)}[f(\tilde X_t)]d{\rm M}(y,\xi)=
\int_\hir  f(y,\xi)d{\rm M}(y,\xi)$$
(resp. $\int_\hir  \Ebb_{\P}^{(y,\xi)}[f(X_t)]d{\rm M}(y,\xi)=
\int_\hir  f(y,\xi)d{\rm M}(y,\xi)$
)
 for every $f\in L^1({\rm M})$. 
 \end{lemma}
\proof
Since $0$ is the lowest eigenvalue of $\LN$, we have 
$$\int_\hir  \Ebb_{\P}^{(y,\xi)}[f(\tilde X_t)]d{\rm M}(y,\xi)=
(\one, e^{-t\LN}f)_{\LM}=
(e^{-t\LN}\one, f)_{\LM}=
(\one, f)_{\LM}=
\int_\hir  f((y,\xi)) d{\rm M}(y,\xi).$$
Similarly we can see that 
$\int_\hir  \Ebb_{\P}^{(y,\xi)}[f(X_t)]d{\rm M}(y,\xi)=
\int_\hir  f((y,\xi)) d{\rm M}(y,\xi)$. 
This implies the lemma. 
\qed

A FCLT for the process $\proo F$ under $\P$ is now immediate. 
Using Proposition \ref{kv}, we have the following result for both the relativistic and non-relativistic  Nelson models.
Now we can state the second main theorem in this paper. 
\begin{theorem}
\label{main5}
Suppose that $f\in D(\LN^{-1})$ and $\int_\hir  f d{\rm M}=0$.
Let 
$\SS(f)=2 \| \LN^{-1/2} f \|_{\LM}^2$.
Then the following holds:
\bi
\item
{\rm (Relativistic case)}
$\d \lim_{s\to\infty}\frac{1}{\sqrt s} \FF{st}
=\SS(f) B_t$.
\item
{\rm (Non-relativistic case)}
$\d \lim_{s\to\infty}\frac{1}{\sqrt s} \FFF{st}
=\SS(f) B_t$.
\ei
Here the convergence is in the weak sense. 
\end{theorem}
\proof
We give the proof only for the relativistic case. The proof for non-relativistic case is similar. 
Suppose that 
$f\in D(\LN^{-1/2})$ is  $\RR$-valued. 
$\proo {\tilde X}$ is Markov process with positive self-adjoint generator $\LN$.  
By Lemma~\ref{denden} 
$e^{-t\LN}$ for every $t\geqslant0$ is positivity improving. Hence, by Remark \ref{ergo}, the random process $\proo X$ is ergodic. 
By Lemma \ref{inv}, the probability measure $\rm M$ is an  invariant measure. 
The stationarity follows from the reflection symmetry and shift invariance of $P(\phi)_1$ processes.  Thus all assumptions of Proposition \ref{kv} are verified and we deduce that $ (\FF{t})_{t\geqslant 0}$ satisfies a functional central limit theorem relative to 
$\tilde \P$, and the variance is given by $-2 \| \LN^{-1/2} f \|_{\LM}^2$. \qed

\subsection{Examples of diffusion constants}
Here we give some examples of functions 
$f\in D(\LN ^{-1/2})$, 
$f:\hir\ni(x,\xi)\mapsto f(x,\xi) \in \mathbb C$, and $\Ebb_{\rm M}[f]=0$ of direct
interest in the FCLT. 
For simplicity we assume  
$\omega(k)=\sqrt{|k|^2+\nu^2}$ with $\nu>0$. 
\begin{lemma}\label{domain}
Let 
$f$ be a function on $\hir$ such that $f\gr\in L^2({\rm P})$.
Suppose that $(\gr, f\gr)_{L^2({\rm P})}=0$. Then 
$\Ebb_{\rm M}[f]=0$ and $f\in D(\LN^{-1/2})$, equivalently 
$f\gr\in D((H-E)^{-1/2})$. 
\end{lemma}
\proof
In the case of $\nu>0$ 
it is established in e.g., \cite{bfs2,ge} that 
$\inf \Spec(\LN)=0$ is discrete. 
We set $\inf(\Spec(\LN)\setminus\{0\})=\lambda>0$, 
and it follows that $\Spec(\LN)\subset \{0\}\cup[\lambda,\infty)$. 
Let $E$ be the spectral projection of $\LN$. 
We divide the Hilbert space $\LM $ as 
$$\LM =E(\{0\})\LM \oplus E([\lambda,\infty))\LM .$$
Since the ground state $\one$ of $\LN$ is unique, $E(\{0\})\LM $ is the one dimensional subspace spanned  by $\one$. 
Hence  $(\gr, f\gr)_{L^2({\rm P})}=0$ implies that 
$f\in \LM $ and $(\one, f)_{\LM }=0$. 
We have 
 $f\in  E([\lambda,\infty))\LM $. 
Note that 
$E([\lambda,\infty))\LM \subset D(\LN^{-1/2})$, and $f\in D(\LN^{-1/2})$ 
follows. Finally 
$\Ebb_{\rm M}[f]=(\one, f)_{\LM }=
(\gr, f\gr)_{L^2({\rm P})}=0$. The proof is complete.
\qed

We also introduce the notation 
$$[f]=f-(\gr, f\gr)_\LP.$$

\begin{example}\label{e1}
\rm{
Let $f(x,\xi)=[ \gamma\cdot x]$. 
By Remark \ref{X},  we have $\gr\in D(\gamma\cdot x)$.   
We see that $(\gr, [\gamma\cdot x]\gr)=0$, which implies that 
$\Ebb_{\rm M}[[\gamma\cdot x]]=0$ and $[\gamma\cdot x] \gr\in D(\LN^{-1/2})$ by Lemma \ref{domain}. 
We get 
$$\SS([\gamma\cdot x])=\|(H-E)^{-1/2} [\gamma\cdot x]\gr\|^2.$$
}
\end{example}

\begin{example}\label{e2}
\rm{
Let $f(x,\xi)= \gamma\cdot x$. 
Let $R_pF(x)=F(-x)$ for $F\in \LM$, and $r_fg(k)=g(-k)$. 
The second quantization of $r_f$ is denoted by $\Gamma(r_f)$. 
Set $R=R_p\otimes \Gamma(r_f)$. 
Suppose that $V(x)=V(-x)$ and $\vp(-k)=\vp(k)$. 
Then $\LN RF=R\LN F$ for $F\in D(\LN)$. 
We can see that $\gr$ also has reflection symmetry, i.e., $R\gr=\gr$. 
The reflection symmetry implies that 
$$(\gr, \gamma\cdot x \gr)_{\LP}=(R\gr, R\gamma\cdot x \gr)_{\LP}=
-(\gr, \gamma\cdot x \gr)_{\LP}=0.$$ 
We get $$\SS(\gamma\cdot x)=\|(H-E)^{-1/2} \gamma\cdot x\gr\|^2.$$
}
\end{example}

\begin{example}\label{enew}
\rm{
Let $R$ be the reflection defined in Example \ref{e2}. 
Let $P_f=\int k\add(k) a(k) dk$ be the field momentum.
Suppose that $V(x)=V(-x)$ and $\vp(-k)=\vp(k)$. 
The reflection symmetry implies that 
$$(\gr, \gamma\cdot P_f \gr)_{\LP}=(R\gr, R\gamma\cdot P_f \gr)_{\LP}=
-(\gr, \gamma\cdot P_f \gr)_{\LP}=0.$$ 
We get $$\SS(\gamma\cdot P_f)=\|(H-E)^{-1/2} \gamma\cdot P_f\gr\|^2.$$
A related example is given in \cite{BS}}
\end{example}

\begin{example}\label{e3}
\rm{
Let $f(x,\xi)=\xi(h)$ with $\|h\|_\LR<\infty$ and 
$\|h/\sqrt{\omega}\|_\LR<\infty$. 
We have $\gr\in D(\xi(h))$ by 
the estimate $\|\xi(h) F\|\leq (\|h/\sqrt\omega\|+\|h\|)\|(\hf+\one)^\han F\|$. 
We see that $(\gr, [\xi(h)]\gr)=0$, which implies that 
$\Ebb_{\rm M}[[\xi(h)]]=0$ and $[\xi(h)] \gr\in D(\LN^{-1/2})$. 
We get 
$$\SS([\xi(h)])=\|(H-E)^{-1/2} [\xi(h)]\gr\|^2.$$
}\end{example}

\begin{example}\label{e4}
\rm{
Let $f(x,\xi)=(\gamma \cdot x)\xi(h)$.  We also see that
$\gr\in D((\gamma \cdot x)\xi(h))$ and 
$$\SS([(\gamma \cdot x)\xi(h)])=\|(H-E)^{-1/2} [(\gamma \cdot x)\xi(h)]\gr\|^2.$$
}
\end{example}

\begin{example}\label{e5}
\rm{
Let $f(z,\xi)=e^{z\xi(h)}$, where $z\in\CC$.  
By Remark \ref{NN} we can see that 
$\gr\in D(e^{z\xi(h)})$. 
We get 
$$\SS([e^{z\xi(h)}])=\|(H-E)^{-1/2} [e^{z\xi(h)}]\gr\|^2.$$
}
\end{example}

\begin{example}\label{e6}
\rm{
Consider 
non-relativistic  or 
massive relativistic Nelson Hamiltonians.
Let $f(x,\xi)=e^{\eps|x|}$, where $\eps>0$ is sufficiently small. 
By Remark \ref{X},  we see that 
$\gr\in D(e^{\eps |x|})$. 
We get 
$$\SS([e^{\eps |x|}])=\|(H-E)^{-1/2} [e^{\eps |x|}]\gr\|^2.$$
Furthermore under the same assumptions on $V$ and $\vp$ in Example \ref{e2} 
we see that $\gr\in D(\gamma\cdot x e^{\eps|x|})$ and 
$(\gr, \gamma\cdot x e^{\eps|x|}\gr)_\LP=0$ by the reflection symmetry. 
We get 
$$\SS(\gamma\cdot x e^{\eps|x|})=\|(H-E)^{-1/2} \gamma\cdot x e^{\eps|x|}\gr\|^2.$$
}
\end{example}

\noindent
{\bf Acknowledgements:}
FH is financially supported by JSPS Open Partnership Joint Projects between Japan and Tunisia
``Non-commutative infinite dimensional harmonic analysis: A unified approach from representation theory and probability
theory", and is also supported by
Grant-in-Aid for Scientific Research (B)16H03942
and 
Grant-in-Aid for Scientific Research (B)20H01808
from JSPS.

\bibliographystyle{amsplain}

\end{document}